\documentclass[a4paper, onecolumn]{quantumarticle}
\pdfoutput=1

\usepackage{chngcntr}
\counterwithin{figure}{section}

\usepackage{graphicx}
\graphicspath{ {./img/} }

\usepackage{tikz}
\usepackage{amsmath}
\usepackage{hyperref}
\usepackage{amsfonts}
\usepackage{amsthm}
\usepackage{amssymb}
\usepackage{enumerate}
\usepackage{float}
\usepackage{stmaryrd}
\usepackage{mathtools}
\usepackage{tikzit}

\tikzstyle{red dot}=[fill={rgb,255: red,232; green,165; blue,165}, draw=black, shape=circle, tikzit category=ZX, tikzit fill=red]
\tikzstyle{green dot}=[fill={rgb,255: red,216; green,248; blue,216}, draw=black, shape=circle, tikzit category=ZX, tikzit fill=green]
\tikzstyle{yellow H}=[fill=yellow, draw=black, shape=rectangle, tikzit category=ZX]
\tikzstyle{white dot}=[fill=white, draw=black, shape=circle, tikzit category=ZH]
\tikzstyle{white box}=[fill=white, draw=black, shape=rectangle, tikzit category=ZH]
\tikzstyle{dark}=[fill={rgb,255: red,200; green,200; blue,200}, draw=black, shape=circle, tikzit category=ZH, tikzit fill=gray]
\tikzstyle{halfscalar}=[star, fill=black, draw=black, minimum size=8pt, inner sep=0pt]
\tikzstyle{multiline text}=[align=center]
\tikzstyle{smallblack}=[fill=black, draw=black, shape=circle, inner sep=0pt, minimum size=3pt]

\tikzstyle{blue dashed}=[-, draw=blue, densely dashed]
\tikzstyle{red dashed}=[-, draw=red, densely dashed]
\tikzstyle{purple dashed}=[-, draw=purple, densely dashed]
\tikzstyle{empty}=[-, fill={rgb,255: red,191; green,191; blue,191}, draw={rgb,255: red,191; green,191; blue,191}]
\tikzstyle{green}=[-, draw=green]
\tikzstyle{brace edge}=[-, tikzit draw=blue, decorate, decoration={brace,amplitude=1mm,raise=-1mm}]
\usetikzlibrary{quantikz2}
\usepackage{subfig}
\usepackage{multirow, tabularx}
\usepackage{adjustbox}
\usepackage[section]{placeins}

\newcommand{\doubleZ}{\mathbb{Z}}

\newcommand{\RM}{\mathrm}
\newcommand{\counting}[1][]{\mathord{\#}\mathrm #1}

\newcommand{\npsharpp}{\RM{NP}^{\counting{P}}}
\newcommand{\npcep}{\RM{NP}^{\RM{C_=P}[1]}}
\newcommand{\cep}{\RM{C_=P}}
\newcommand{\paths}{\counting{Paths}}
\newcommand{\apaths}{\counting{AccPaths}}
\newcommand{\rpaths}{\counting{RejPaths}}

\newcommand{\problem}[1]{\mathbf{#1}}
\newcommand{\StateEq}{\problem{StateEq}}
\newcommand{\ContainsEntry}{\problem{ContainsEntry}}
\newcommand{\ContainsEntryk}{\problem{ContainsEntry_k}}
\newcommand{\CircuitExtraction}{\problem{CircuitExtraction}}
\newcommand{\SAT}{\problem{SAT}}
\newcommand{\cSAT}{\problem{\#SAT}} 
\newcommand{\ComparecSAT}{\problem{Compare\#SAT}}
\newcommand{\ScalarDiagram}{\problem{ScalarDiagram}}
\newcommand{\SATandComparecSAT}{\problem{SAT\&Compare\#SAT}}
\newcommand{\CompareDiagrams}{\problem{CompareDiagrams}}
\newcommand{\IsZero}{\problem{IsZero}}

\newcommand{\csat}{\problem{\#SAT}} 

\newcommand{\diageval}[1][]{\left\llbracket #1 \right\rrbracket}

\newcommand{\Rho}{P}

\newcommand{\wedgevee}{%
  \mathbin{{\wedge}\mkern-5mu{\vee}}%
}

\usepackage{xcolor}

\theoremstyle{plain}\newtheorem{theorem}{Theorem}[section]
\theoremstyle{plain}\newtheorem{corollary}[theorem]{Corollary}
\theoremstyle{plain}\newtheorem{lemma}[theorem]{Lemma}
\theoremstyle{definition}\newtheorem{definition}[theorem]{Definition}
\theoremstyle{plain}
\theoremstyle{plain}
\theoremstyle{definition}
\theoremstyle{definition}
\theoremstyle{definition}\newtheorem{example}[theorem]{Example}

\providecommand{\keywords}[1]
{
  \textbf{Keywords: } #1
}

\begin{document}
\title{Constructing $\mathrm{NP}^{\mathord{\#}\mathrm P}$-complete problems and \mbox{${\mathord{\#}\mathrm P}$-hardness} of circuit extraction in phase-free ZH}

\author{Piotr Mitosek}
\affiliation{School of Computer Science, University of Birmingham}
\email{pbm148@student.bham.ac.uk}

\begin{abstract}
\sloppy The ZH calculus is a graphical language for quantum computation reasoning.
The phase-free variant offers a simple set of generators that guarantee universality.
ZH calculus is effective in MBQC and analysis of quantum circuits constructed with the universal gate set Toffoli+H.
While circuits naturally translate to ZH diagrams, finding an ancilla-free circuit equivalent to a given diagram is hard.
Here, we show that circuit extraction for phase-free ZH calculus is $\RM{\#P}$-hard, extending the existing result for ZX calculus.
Another problem believed to be hard is comparing whether two diagrams represent the same process.
We show that two closely related problems are \mbox{$\mathrm{NP}^{\mathord{\#}\mathrm P}$-complete}.
The first problem is: given two processes represented as diagrams, determine the existence of a computational basis state on which they equalize.
The second problem is checking whether the matrix representation of a given diagram contains an entry equal to a given number.
Our proof adapts the proof of Cook-Levin theorem to a reduction from a non-deterministic Turing Machine with access to ${\mathord{\#}\mathrm P}$ oracle.
\end{abstract}



\section{Introduction}
Graphical representation of the quantum processes via string diagrams \cite{coeckePicturingQuantumProcesses2017} is a quickly developing area with a new view on the most important problems of theoretical quantum computation. String diagrams offer a concise and precise representation of quantum processes and transformations they can undergo \cite{coeckeInteractingQuantumObservables2011}. Graphical calculi such as ZX calculus \cite{coeckeInteractingQuantumObservables2011, vandeweteringZXcalculusWorkingQuantum2020} have been successfully used to tackle the circuit optimization \cite{duncanGraphtheoreticSimplificationQuantum2020, staudacherReducing2QuBitGate2023}, error correction \cite{duncanVerifyingSteaneCode2014}, and many other problems. ZH calculus, introduced in \cite{backensZHCompleteGraphical2019}, has been applied when reasoning about measurement-based quantum computation \cite{kupperAnalysisQuantumHypergraph2020} and circuits constructed with the Toffoli gate \cite{kuijpersGraphicalFourierTheory2019}. Normally, ZH features generators with a phase label that ranges from $0$ to $2\pi$. The phase-free variant \cite{backensCompletenessZHcalculus2023} offers a simple set of generators that guarantee universality and closely relate to circuits constructed with the Toffoli+H universal gate set, without arbitrary phases in spider generators.

The standard notation of quantum computation uses circuits. Rewrite rules for quantum circuits exist (for example, see \cite{clementCompleteEquationalTheory2023a}), but can get very complicated, while diagram rewriting rules are usually compact. In graphical calculi, a few generators are sufficient for universality, while the addition of appropriate rewrite rules may offer completeness for various parts of quantum mechanics \cite{jeandelCompleteAxiomatisationZXCalculus2018a, backensCompletenessZXcalculus2016, ngUniversalCompletionZXcalculus2017, backensCompletenessZHcalculus2023}. Thus, one can often benefit from the translation of the quantum process to a string diagram. Typical uses of graphical calculi follow the same pattern. First, a circuit or other presentation of quantum computation like an MBQC scheme is translated to a graphical calculus. Next, the obtained diagram is modified via rewrite rules. Finally, a circuit is extracted from the new diagram. The circuit extraction, in general, is $\counting{P}$-hard \cite{debeaudrapCircuitExtractionZXDiagrams2022c}. The existence of various flow structures guarantees fast extraction \cite{duncanGraphtheoreticSimplificationQuantum2020, backensThereBackAgain2021, simmonsRelatingMeasurementPatterns2021}. Sometimes these properties can be preserved while rewriting diagrams \cite{mcelvanneyCompleteFlowPreservingRewrite2023}. Circuit extraction is vital, as circuits are the machine language for most existing physical quantum computers.

The most natural problem to ask about any model of computation is to say whether two computation schemes are equivalent. For graphical calculi, this problem is comparing diagrams:

\begin{quote}
    $\CompareDiagrams$\\
    \textbf{Input:} two phase-free ZH diagrams $D_1$ and $D_2$.\\
    \textbf{Output:} $True$ if and only if $\diageval[D_1] = \diageval[D_2]$.
\end{quote}

One of the ways to check whether two diagrams represent the same quantum process is to attempt rewriting one to the other. Completeness of graphical calculi means, that if two diagrams represent the same process, one can be transformed to the other by a series of rewrites. For languages complete for quantum computation, the proof goes through exponential in size normal forms. It means that, concerning the current knowledge, the number of rewrites transforming one diagram to another not only can be exponential, but the diagrams on the way can be exponential as well. The best known upper bound for comparing diagrams is $\rm{coNP}^{\counting{P}}$ \cite{debeaudrapCircuitExtractionZXDiagrams2022c}.

The goal of the research is to establish a tight bound for comparing diagrams. The main motivations come from \cite{debeaudrapCircuitExtractionZXDiagrams2022c}, where an upper bound for comparing diagrams is used to bound circuit extraction from above with $\RM{NP}^{\npsharpp}$. Further, the hardness of comparing diagrams would extend to the comparison of circuits with measurement post-selection and thus, it would connect to $\RM{PostBQP}$ (Bounded-error Quantum Polynomial-time with Post-selection \cite{aaronsonQuantumComputingPostselection2005}) and $\RM{PP}$ (Probabilistic Polynomial-time \cite{gillComputationalComplexityProbabilistic1977}) computation schemes. Another motivation was exploring the encoding of $\counting{P}$-hard problems as ZH diagrams, and to look for interesting completeness results.

In this work, we make progress towards the above goal. We show that two closely related problems, arising in phase-free ZH calculus \cite{backensCompletenessZHcalculus2023}, are $\npsharpp$-complete. These are some of the first examples of $\npsharpp$-complete problems (other examples that do not directly relate to quantum can be found in \cite{toranComplexityClassesDefined1991, monniauxNPExistsPP2022}). The problems relate to comparing diagrams in the following way. Comparing diagrams asks that a certain property holds for all entries in the matrix representation, while the presented problems check whether there exists an entry satisfying the property.

The first of the problem mentioned above is:

\begin{quote}
    $\StateEq$\\
    \textbf{Input:} a pair of phase-free ZH diagrams $D_1, D_2$ with matching number of dangling edges.\\
    \textbf{Output:} $True$ if and only if there exists a computational basis state $\ket{v}$ such that $\diageval[D_1] \ket{v} = \diageval[D_2] \ket{v}$.
\end{quote}

By interpreting a diagram with dangling edges as a quantum state, this problem can be viewed as checking whether there exists a choice of measurement outcomes (on the computational basis) that appears with the same probability for both quantum states given as diagrams. Another way to understand this problem is by checking whether matrix interpretations of given diagrams are equal in some positions.

The second problem is as follows, for $k \in \mathbb{Z}[\frac{1}{2}]$ (dyadic rationals):

\begin{quote}
    $\ContainsEntryk$\\
    \textbf{Input:} a phase-free ZH diagram $D$\\
    \textbf{Output:} $True$ if and only if $k$ appears in the matrix interpretation of $D$.
\end{quote}

In particular, when $k = 0$, the problem can be interpreted as: given a diagram representing a quantum state, does there exist a measurement choice (on the computational basis) that happens with probability $0$? Our result also works when the number $k$ is part of the input, rather than a fixed number on which the problem depends. Our main contribution is proving the following:

\begin{theorem}\label{main result}
    $\StateEq$ and $\ContainsEntryk$ are both $\npsharpp$-complete.
\end{theorem}

We also give a construction showing $\counting{P}$-hardness of circuit extraction from phase-free ZH -- the original paper \cite{debeaudrapCircuitExtractionZXDiagrams2022c} showing $\counting{P}$-hardness of circuit extraction did not work for the phase-free variant -- the construction used in the proof required $\frac{\pi}{2}$ phase spider, which does not exist in phase-free ZH.

\begin{quote}
    $\CircuitExtraction$\\
    \textbf{Input:} a phase-free ZH diagram $D_1$ proportional to a unitary and a (finitely presented) set of quantum gates $\mathcal{G}$.\\
    \textbf{Output:} a polynomial in size of $D_1$ quantum circuit constructed with gates from $\mathcal{G}$ representing a process proportional to $\diageval[D_1]$ or a message that no such circuit exists.
\end{quote}

\begin{theorem}\label{circuit extraction hardness}
    $\CircuitExtraction$ is $\counting{P}$-hard.
\end{theorem}

Our proofs work for phase-free ZH calculus. These results also generalise to other graphical calculi, like ZX, ZH and ZW. This is because the phase-free ZH diagrams can always be represented in those other graphical calculi, and the translation from phase-free ZH to other calculi is polynomial in size.

\subsection{Structure}
In the section \ref{coco}, we present the necessary definitions from the Computational Complexity. After that, in section \ref{phasefreezh}, we define the phase-free ZH calculus. The proof of theorem \ref{main result} follows in section \ref{npsharppsection}. In section \ref{circuitextractionsection}, we give the proof of theorem \ref{circuit extraction hardness}. In the last section \ref{lastsection}, we give conclusions and plans for further research.

\section{Computational Complexity}\label{coco}
We assume knowledge of Turing Machines and familiarity with standard classes such as $\RM{NP}$ and the concept of oracle classes. The names of problems are written in bold text, while the names of classes are in standard, non-italic text. TM stands for Turing Machine and ND for non-deterministic. $\paths(M,w)$ for a NDTM $M$ and its input $w \in \Sigma^*$ stands for the total number of paths $M$ runs on $w$. Similarly, $\apaths(M,w)$ and $\rpaths(M,w)$ stand for the total number of respectively accepting and rejecting paths $M$ has on $w$. $q_{ACC}$ and $q_{REJ}$ stand for the accepting and rejecting states of a TM respectively. Finally, $q_{ASK}$ stands for the state used to communicate with the oracle, i.e. when an oracle TM enters its $q_{ASK}$ state, then the oracle is called.

Some of the classes we work with are non-standard. For clarity, in the appendix \ref{cococlasses}, we present all of the complexity classes that we use or mention.

We are going to talk about boolean formulae and their valuations. Boolean formulae are constructed with variables, two logical values $True$ and $False$, and connectives $\wedge, \vee, \neg, \to, \leftrightarrow$. A valuation is a function mapping some variables to logical values $True$ and $False$. For instance, $(x_1 \vee x_2)$ under the valuation $(x_1, False), (x_2, True)$ is $(False \vee True) = True$, and under the valuation $(x_1, False)$ it is $(False \vee x_2) = x_2$.

Throughout, we use the following two problems, based on $\SAT$ problem, the canonical $\RM{NP}$-complete problem \cite{cookComplexityTheoremprovingProcedures1971}.

\begin{quote}
    $\cSAT$\\
    \textbf{Input:} a boolean formula $\phi$ and the variables $x_1,\dots,x_n$ on which $\phi$ is defined.\\
    \textbf{Output:} a number of satisfying assignments of $\phi$, i.e. number of $True/False$ assignments to $x_1,\dots,x_n$ for which $\phi$ evaluates to $True$.
\end{quote}

We will also use $\csat$ as the corresponding function mapping the formula and variables to the number of satisfying assignments. When the variables are explicit, for instance, all variables appearing in the formula and nothing else, we omit writing them. For instance: {{$\csat((x_1 \wedge x_2) \wedge (x_1 \wedge \neg x_3)) = 1$}} as the only satisfying assignment of $(x_1, x_2, x_3)$ is $(True, True, False)$ The $\csat$ problem is the canonical $\counting{P}$-complete problem \cite{valiantComplexityComputingPermanent1979}.

\begin{quote}
    $\ComparecSAT$\\
    \textbf{Input:} Two boolean formulae defined on $n$ variables each:\begin{itemize}
        \item $\phi$ on variables $X := x_1, \dots, x_n$,
        \item $\psi$ on variables $Y := y_1, \dots, y_n$.
    \end{itemize}
    \textbf{Output:} $True$ if and only if $\csat(\phi, X) = \csat(\psi, Y)$.
\end{quote}

Similarly, we omit writing variables when they are explicit. For example, the answer for instance consisting of formulae $((x_1 \wedge x_2) \wedge (x_1 \wedge \neg x_3)$ and $(y_1 \vee y_2 \vee y_3)$ is $False$, as $\csat((x_1 \wedge x_2) \wedge (x_1 \wedge \neg x_3)) = 1 \ne 7 = \csat(y_1 \vee y_2 \vee y_3)$. This problem is $\cep$-complete. While the proof follows immediately from the definitions of $\cep$ and $\counting{P}$, we present it below, as we could not find any publications with the proof (or this problem definition).

\begin{theorem}
    $\ComparecSAT$ is $\cep$-complete.
\end{theorem}
\begin{proof}
    \textbf{Containment:}\\
    Define a NDTM that given an input $(\phi,\psi)$, where the formulae are on $n$ variables each, runs $2^{n+1}$ paths. It uses $2^n$ paths to evaluate all possible assignments of $\phi$ and returns $True$ precisely for satisfying assignments of $\phi$. Similarly, on the other $2^n$ paths it evaluates and returns $True$ precisely for unsatisfying assignments of $\psi$. Thus, in total $M$ accepts input on $\csat(\phi) + \csat(\neg\psi) = \csat(\phi) + 2^n - \csat(\psi)$ paths and rejects on $2^n-\csat(\phi) + \csat(\psi)$ paths. The two numbers equal precisely when $\csat(\phi) = \csat(\psi)$, as required.

    \textbf{Hardness:}\\
    Given a problem $A \in \RM{C_=P}$, let $M$ be the corresponding NDTM. For any instance $w$ for $A$, we use the Cook-Levin method \cite{cookComplexityTheoremprovingProcedures1971} to construct a boolean formula $\phi$ corresponding to running $M$ on $w$. This construction is parsimonious (or at least it can be made parsimonious, check \cite{aroraComputationalComplexityModern2009}), thus $\csat(\phi) = \apaths(M,w)$. Similarly, we construct a boolean formula $\psi$ corresponding to running $M'$ on $w$, where $M'$ is $M$ with accepting and rejecting states flipped. Then, again by parsimony of the Cook-Levin method, $\csat(\psi) = \rpaths(M,w)$. We extend both formulae to $\phi', \psi'$ that operate on the same number of variables, without changing the number of satisfying assignments (for instance, by appending a clause $(y_1 \wedge y_2 \wedge \dots \wedge y_l)$ where $y_1, \dots, y_l$ are newly introduced variables). Finally, the instance $w$ of $A$ is equivalent to $(\psi', \phi')$ instance of $\ComparecSAT$, as $\csat(\phi') = \csat(\psi') \Leftrightarrow \apaths(M,w) = \rpaths(M,w)$.
\end{proof}

The trouble of proving $\npsharpp$-hardness of $\StateEq$ and $\ContainsEntryk$ comes from the lack of a simple $\npsharpp$-complete problem that can be used in reductions. Crafting such a problem is the most difficult part and we will dedicate it most of the proof section.

\section{Phase-free ZH Calculus}\label{phasefreezh}
ZH calculus is a graphical language for quantum computation. In this work, we use phase-free ZH \cite{backensCompletenessZHcalculus2023}. The language consists of string diagrams and rules describing how to modify the diagrams.

The string diagrams represent quantum processes given by matrices over $\doubleZ[\frac{1}{2}]$ (dyadic rationals), corresponding to the Toffoli+Hadamard approximately universal gate set.

The diagrams can be composed sequentially or parallelly, corresponding to the composition or the tensor product. The generators (including derived generators) of phase-free ZH, together with matrix interpretations, are presented in figure \ref{zh_generators}. Double square brackets stand for matrix interpretation. Note that the star generator is just a scalar. The matrix interpretation of a box contains ones everywhere except for the bottom right corner, where it instead has a $-1$.

\begin{figure}
\centering
\begin{align*}
\left \llbracket \begin{tikzpicture}
	\begin{pgfonlayer}{nodelayer}
		\node [style=white dot] (0) at (0, 0) {};
		\node [style=none] (1) at (-2, 1) {};
		\node [style=none] (2) at (2, 1) {};
		\node [style=none] (3) at (-2, -1) {};
		\node [style=none] (4) at (2, -1) {};
		\node [style=none] (5) at (-3, 0) {$n$};
		\node [style=none] (6) at (3, 0) {$m$};
		\node [style=none] (7) at (-1.5, 0.25) {$\vdots$};
		\node [style=none] (8) at (1.5, 0.25) {$\vdots$};
		\node [style=none] (9) at (2.5, 1.25) {};
		\node [style=none] (10) at (2.5, -1.25) {};
		\node [style=none] (11) at (-2.5, -1.25) {};
		\node [style=none] (12) at (-2.5, 1.25) {};
	\end{pgfonlayer}
	\begin{pgfonlayer}{edgelayer}
		\draw [bend left] (1.center) to (0);
		\draw [bend left] (0) to (2.center);
		\draw [bend right] (0) to (4.center);
		\draw [bend left] (0) to (3.center);
		\draw [style=brace edge] (9.center) to (10.center);
		\draw [style=brace edge] (11.center) to (12.center);
	\end{pgfonlayer}
\end{tikzpicture} \right \rrbracket &= \ket{0}^{\otimes n}\bra{0}^{\otimes m} + \ket{1}^{\otimes n}\bra{1}^{\otimes m}\\
\left \llbracket \begin{tikzpicture}
	\begin{pgfonlayer}{nodelayer}
		\node [style=dark] (9) at (0, 0) {};
		\node [style=none] (10) at (-2, 1) {};
		\node [style=none] (11) at (2, 1) {};
		\node [style=none] (12) at (-2, -1) {};
		\node [style=none] (13) at (2, -1) {};
		\node [style=none] (14) at (-3, 0) {$n$};
		\node [style=none] (15) at (3, 0) {$m$};
		\node [style=none] (16) at (-1.5, 0.25) {$\vdots$};
		\node [style=none] (17) at (1.5, 0.25) {$\vdots$};
		\node [style=none] (18) at (2.5, 1.25) {};
		\node [style=none] (19) at (2.5, -1.25) {};
		\node [style=none] (20) at (-2.5, -1.25) {};
		\node [style=none] (21) at (-2.5, 1.25) {};
	\end{pgfonlayer}
	\begin{pgfonlayer}{edgelayer}
		\draw [bend left] (10.center) to (9);
		\draw [bend left] (9) to (11.center);
		\draw [bend right] (9) to (13.center);
		\draw [bend left] (9) to (12.center);
		\draw [style=brace edge] (18.center) to (19.center);
		\draw [style=brace edge] (20.center) to (21.center);
	\end{pgfonlayer}
\end{tikzpicture} \right \rrbracket &= \sqrt{2}(\ket{+}^{\otimes n}\bra{+}^{\otimes m} + \ket{-}^{\otimes n}\bra{-}^{\otimes m})\\
\left \llbracket \begin{tikzpicture}
	\begin{pgfonlayer}{nodelayer}
		\node [style=white dot] (10) at (0, 0) {};
		\node [style=none] (11) at (-2, 1) {};
		\node [style=none] (12) at (2, 1) {};
		\node [style=none] (13) at (-2, -1) {};
		\node [style=none] (14) at (2, -1) {};
		\node [style=none] (15) at (-3, 0) {$n$};
		\node [style=none] (16) at (3, 0) {$m$};
		\node [style=none] (17) at (-1.5, 0.25) {$\vdots$};
		\node [style=none] (18) at (1.5, 0.25) {$\vdots$};
		\node [style=none] (19) at (2.5, 1.25) {};
		\node [style=none] (20) at (2.5, -1.25) {};
		\node [style=none] (21) at (-2.5, -1.25) {};
		\node [style=none] (22) at (-2.5, 1.25) {};
		\node [style=none] (23) at (0, 0) {$\neg$};
	\end{pgfonlayer}
	\begin{pgfonlayer}{edgelayer}
		\draw [bend left] (11.center) to (10);
		\draw [bend left] (10) to (12.center);
		\draw [bend right] (10) to (14.center);
		\draw [bend left] (10) to (13.center);
		\draw [style=brace edge] (19.center) to (20.center);
		\draw [style=brace edge] (21.center) to (22.center);
	\end{pgfonlayer}
\end{tikzpicture} \right \rrbracket &= \ket{0}^{\otimes n}\bra{0}^{\otimes m} - \ket{1}^{\otimes n}\bra{1}^{\otimes m}\\
\left \llbracket \begin{tikzpicture}
	\begin{pgfonlayer}{nodelayer}
		\node [style=dark] (10) at (0, 0) {};
		\node [style=none] (11) at (-2, 1) {};
		\node [style=none] (12) at (2, 1) {};
		\node [style=none] (13) at (-2, -1) {};
		\node [style=none] (14) at (2, -1) {};
		\node [style=none] (15) at (-3, 0) {$n$};
		\node [style=none] (16) at (3, 0) {$m$};
		\node [style=none] (17) at (-1.5, 0.25) {$\vdots$};
		\node [style=none] (18) at (1.5, 0.25) {$\vdots$};
		\node [style=none] (19) at (2.5, 1.25) {};
		\node [style=none] (20) at (2.5, -1.25) {};
		\node [style=none] (21) at (-2.5, -1.25) {};
		\node [style=none] (22) at (-2.5, 1.25) {};
		\node [style=none] (23) at (0, 0) {$\neg$};
	\end{pgfonlayer}
	\begin{pgfonlayer}{edgelayer}
		\draw [bend left] (11.center) to (10);
		\draw [bend left] (10) to (12.center);
		\draw [bend right] (10) to (14.center);
		\draw [bend left] (10) to (13.center);
		\draw [style=brace edge] (19.center) to (20.center);
		\draw [style=brace edge] (21.center) to (22.center);
	\end{pgfonlayer}
\end{tikzpicture} \right \rrbracket &= \sqrt{2}(\ket{+}^{\otimes n}\bra{+}^{\otimes m} - \ket{-}^{\otimes n}\bra{-}^{\otimes m})\\
\left \llbracket \begin{tikzpicture}
	\begin{pgfonlayer}{nodelayer}
		\node [style=white box] (9) at (0, 0) {};
		\node [style=none] (10) at (-2, 1) {};
		\node [style=none] (11) at (2, 1) {};
		\node [style=none] (12) at (-2, -1) {};
		\node [style=none] (13) at (2, -1) {};
		\node [style=none] (14) at (-3, 0) {$n$};
		\node [style=none] (15) at (3, 0) {$m$};
		\node [style=none] (16) at (-1.5, 0.25) {$\vdots$};
		\node [style=none] (17) at (1.5, 0.25) {$\vdots$};
		\node [style=none] (18) at (2.5, 1.25) {};
		\node [style=none] (19) at (2.5, -1.25) {};
		\node [style=none] (20) at (-2.5, -1.25) {};
		\node [style=none] (21) at (-2.5, 1.25) {};
	\end{pgfonlayer}
	\begin{pgfonlayer}{edgelayer}
		\draw [bend left] (10.center) to (9);
		\draw [bend left] (9) to (11.center);
		\draw [bend right] (9) to (13.center);
		\draw [bend left] (9) to (12.center);
		\draw [style=brace edge] (18.center) to (19.center);
		\draw [style=brace edge] (20.center) to (21.center);
	\end{pgfonlayer}
\end{tikzpicture} \right \rrbracket &= \sum_{i_1, \dots, i_m, j_1, \dots, j_n \in \{ 0, 1 \}} (-1)^{i_1 \dots i_m j_1 \dots j_n} \ket{j_1 \dots j_n}\bra{i_1 \dots i_m}\\
\left \llbracket \begin{tikzpicture}
	\begin{pgfonlayer}{nodelayer}
		\node [style=halfscalar] (0) at (0, 0) {};
	\end{pgfonlayer}
\end{tikzpicture} \right \rrbracket &= \frac{1}{2}
\end{align*}
\caption{Generators of phase-free ZH calculus. We refer to these as \textbf{white spider}, \textbf{dark spider}, \textbf{white not}, \textbf{dark not}, \textbf{box} and \textbf{star} respectively.}
\label{zh_generators}
\end{figure}

\subsection{Boolean formulae in graphical calculi}
The two logical values $True$ and $False$ from logic can be represented as $\ket{1}$ and $\ket{0}$ respectively, both logical values at the same time by $\ket{0} + \ket{1}$, and the check of whether the value is $True$ by $\bra{1}$, see figure \ref{logical_values}. The representation of basic logic gates in phase-free ZH on such values is presented in figure \ref{basic_gates}. For correctness, see \cite[Subsection 2.3]{backensCompletenessZHcalculus2023}. From these, we can construct arbitrary logic operators -- for instance to represent OR gate we can note that $\phi \vee \psi \Leftrightarrow \neg (\neg \phi \wedge \neg \psi)$.

As an example, the construction for the formula {{$(x_1 \wedge x_2) \wedge (x_1 \wedge \neg x_3)$}} is presented in the figure \ref{fig:formula example}.

\begin{figure}[H]
\centering
\begin{tikzpicture}
	\begin{pgfonlayer}{nodelayer}
		\node [style=none] (0) at (0, 0) {$=$};
		\node [style=white box, minimum width=1cm, minimum height=1.5cm] (1) at (-3.5, 0) {TRUE};
		\node [style=none] (3) at (-1.5, 0) {};
		\node [style=none] (5) at (3.25, 0) {};
		\node [style=dark] (6) at (1.5, 0) {};
		\node [style=none] (7) at (1.5, 0) {$\neg$};
		\node [style=none] (9) at (4, 0) {};
		\node [style=none] (10) at (-5, 0) {};
	\end{pgfonlayer}
	\begin{pgfonlayer}{edgelayer}
		\draw (3.center) to (1);
		\draw (6) to (5.center);
	\end{pgfonlayer}
\end{tikzpicture}\quad
\begin{tikzpicture}
	\begin{pgfonlayer}{nodelayer}
		\node [style=none] (0) at (0, 0) {$=$};
		\node [style=white box, minimum width=1cm, minimum height=1.5cm] (1) at (-3.5, 0) {FALSE};
		\node [style=none] (2) at (-1.5, 0) {};
		\node [style=none] (3) at (3.25, 0) {};
		\node [style=dark] (4) at (1.5, 0) {};
		\node [style=none] (5) at (1.5, 0) {};
		\node [style=none] (6) at (4, 0) {};
		\node [style=none] (7) at (-5, 0) {};
	\end{pgfonlayer}
	\begin{pgfonlayer}{edgelayer}
		\draw (2.center) to (1);
		\draw (4) to (3.center);
	\end{pgfonlayer}
\end{tikzpicture}$ $\\$ $\\$ $\\
\begin{tikzpicture}
	\begin{pgfonlayer}{nodelayer}
		\node [style=none] (0) at (0, 0) {$=$};
		\node [style=white box, minimum width=1.5cm, minimum height=1.5cm] (1) at (-3.5, 0) {};
		\node [style=none] (2) at (-1.5, 0) {};
		\node [style=none] (3) at (3.25, 0) {};
		\node [style=white dot] (4) at (1.5, 0) {};
		\node [style=none] (6) at (4, 0) {};
		\node [style=none] (7) at (-5, 0) {};
		\node [style=multiline text] (8) at (-3.5, 0) {TRUE\\and\\FALSE};
	\end{pgfonlayer}
	\begin{pgfonlayer}{edgelayer}
		\draw (2.center) to (1);
		\draw (4) to (3.center);
	\end{pgfonlayer}
\end{tikzpicture}\quad
\begin{tikzpicture}
	\begin{pgfonlayer}{nodelayer}
		\node [style=none] (0) at (0.25, 0) {$=$};
		\node [style=white box, minimum width=1.8cm, minimum height=1.5cm] (1) at (-2.25, 0) {};
		\node [style=none] (2) at (-4.5, 0) {};
		\node [style=none] (3) at (1, 0) {};
		\node [style=dark] (4) at (2.5, 0) {};
		\node [style=none] (5) at (4, 0) {};
		\node [style=none] (6) at (-5, 0) {};
		\node [style=multiline text] (7) at (-2.25, 0) {Is TRUE?};
		\node [style=none] (8) at (2.5, 0) {$\neg$};
	\end{pgfonlayer}
	\begin{pgfonlayer}{edgelayer}
		\draw (2.center) to (1);
		\draw (4) to (3.center);
	\end{pgfonlayer}
\end{tikzpicture}
\caption{Basic logic values.}
\label{logical_values}
\end{figure}
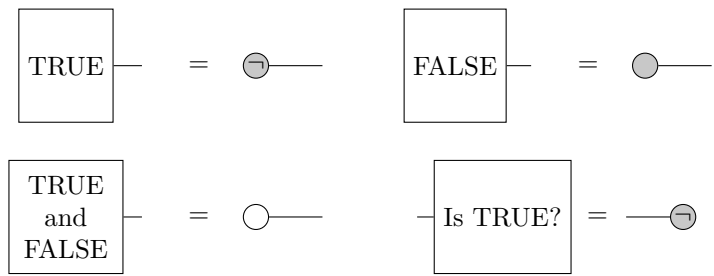

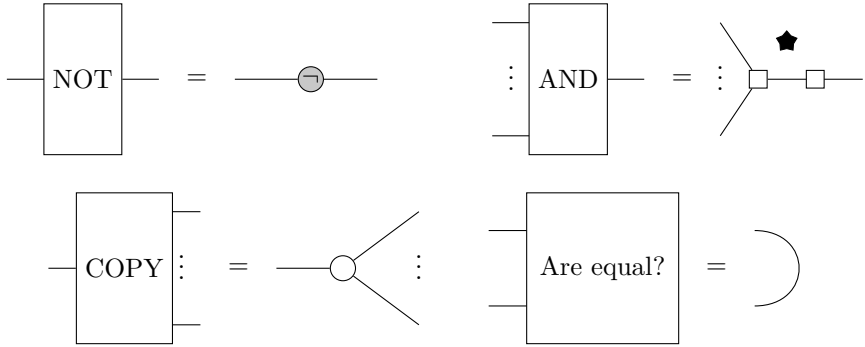
\begin{figure}[H]
\centering
\begin{tikzpicture}
	\begin{pgfonlayer}{nodelayer}
		\node [style=none] (0) at (0, 0) {$=$};
		\node [style=white box, minimum width=1cm, minimum height=2cm] (1) at (-3, 0) {NOT};
		\node [style=none] (2) at (-5, 0) {};
		\node [style=none] (4) at (-1, 0) {};
		\node [style=none] (6) at (1, 0) {};
		\node [style=none] (7) at (4.75, 0) {};
		\node [style=dark] (8) at (3, 0) {};
		\node [style=none] (9) at (3, 0) {$\neg$};
		\node [style=none] (10) at (-3, 0) {};
		\node [style=none] (11) at (6, 0) {};
		\node [style=none] (12) at (-6, 0) {};
	\end{pgfonlayer}
	\begin{pgfonlayer}{edgelayer}
		\draw (4.center) to (1);
		\draw (8) to (7.center);
		\draw (8) to (6.center);
		\draw (2.center) to (10.center);
	\end{pgfonlayer}
\end{tikzpicture}
\begin{tikzpicture}
	\begin{pgfonlayer}{nodelayer}
		\node [style=none] (0) at (0, 0) {$=$};
		\node [style=white box, minimum width=1cm, minimum height=2cm] (1) at (-3, 0) {AND};
		\node [style=none] (2) at (-5, 1.5) {};
		\node [style=none] (3) at (-5, -1.5) {};
		\node [style=none] (4) at (-1, 0) {};
		\node [style=none] (5) at (1, 1.5) {};
		\node [style=none] (6) at (1, -1.5) {};
		\node [style=none] (7) at (4.75, 0) {};
		\node [style=white box] (8) at (2, 0) {};
		\node [style=none] (10) at (-3, 1.5) {};
		\node [style=none] (11) at (-3, -1.5) {};
		\node [style=white box] (13) at (3.5, 0) {};
		\node [style=halfscalar] (14) at (2.75, 1) {};
		\node [style=none] (15) at (-4.5, 0.25) {$\vdots$};
		\node [style=none] (16) at (1, 0.25) {$\vdots$};
		\node [style=none] (17) at (6, 0) {};
		\node [style=none] (18) at (-6, 0) {};
	\end{pgfonlayer}
	\begin{pgfonlayer}{edgelayer}
		\draw (4.center) to (1);
		\draw (5.center) to (8);
		\draw (8) to (7.center);
		\draw (8) to (6.center);
		\draw (3.center) to (11.center);
		\draw (2.center) to (10.center);
	\end{pgfonlayer}
\end{tikzpicture}$ $\\$ $\\$ $\\
\begin{tikzpicture}
	\begin{pgfonlayer}{nodelayer}
		\node [style=none] (0) at (0, 0) {$=$};
		\node [style=white box, minimum width=1cm, minimum height=2cm] (1) at (-3, 0) {COPY};
		\node [style=none] (2) at (-1, 1.5) {};
		\node [style=none] (3) at (-1, -1.5) {};
		\node [style=none] (4) at (-5, 0) {};
		\node [style=none] (5) at (4.75, 1.5) {};
		\node [style=none] (6) at (4.75, -1.5) {};
		\node [style=none] (7) at (1, 0) {};
		\node [style=white dot] (8) at (2.75, 0) {};
		\node [style=none] (10) at (-3, 1.5) {};
		\node [style=none] (11) at (-3, -1.5) {};
		\node [style=none] (12) at (-1.5, 0.25) {$\vdots$};
		\node [style=none] (13) at (4.75, 0.25) {$\vdots$};
	\end{pgfonlayer}
	\begin{pgfonlayer}{edgelayer}
		\draw (4.center) to (1);
		\draw (5.center) to (8);
		\draw (8) to (7.center);
		\draw (8) to (6.center);
		\draw (3.center) to (11.center);
		\draw (2.center) to (10.center);
	\end{pgfonlayer}
\end{tikzpicture}
\begin{tikzpicture}
	\begin{pgfonlayer}{nodelayer}
		\node [style=none] (0) at (0, 0) {$=$};
		\node [style=white box, minimum width=2cm, minimum height=2cm] (1) at (-3, 0) {Are equal?};
		\node [style=none] (2) at (-6, 1) {};
		\node [style=none] (5) at (3, 0) {};
		\node [style=none] (10) at (-6, -1) {};
		\node [style=none] (11) at (-4.5, 1) {};
		\node [style=none] (12) at (-4.5, -1) {};
		\node [style=none] (13) at (1, 1) {};
		\node [style=none] (14) at (1, -1) {};
		\node [style=none] (15) at (-7, 0) {};
	\end{pgfonlayer}
	\begin{pgfonlayer}{edgelayer}
		\draw (2.center) to (11.center);
		\draw (12.center) to (10.center);
		\draw [bend left=90, looseness=2.00] (13.center) to (14.center);
	\end{pgfonlayer}
\end{tikzpicture}
\caption{Basic logic gates.}
\label{basic_gates}
\end{figure}

\begin{figure}[H]
    \centering
    \scalebox{0.7}{\begin{tikzpicture}
	\begin{pgfonlayer}{nodelayer}
		\node [style=none] (2) at (-17, 3) {$x_1$};
		\node [style=none] (3) at (-17, 1) {$x_2$};
		\node [style=none] (4) at (-17, -1) {$x_1$};
		\node [style=none] (5) at (-17, -3) {$x_3$};
		\node [style=none] (6) at (-16, 3) {};
		\node [style=none] (7) at (-16, 1) {};
		\node [style=none] (8) at (-16, -1) {};
		\node [style=none] (9) at (-16, -3) {};
		\node [style=white box] (10) at (-14, 2) {AND};
		\node [style=white box] (12) at (-14, -3) {NOT};
		\node [style=white box] (13) at (-11.5, -2) {AND};
		\node [style=white box] (14) at (-9.5, 0) {AND};
		\node [style=none] (15) at (-11.5, 2) {};
		\node [style=none] (16) at (-7.5, 0) {};
		\node [style=none] (17) at (-14, -1) {};
		\node [style=none] (19) at (-5.5, 0) {$\rightarrow$};
		\node [style=none] (20) at (-2, 3) {};
		\node [style=none] (21) at (-2, 1) {};
		\node [style=none] (22) at (-2, -1) {};
		\node [style=none] (23) at (-2, -3) {};
		\node [style=white box] (24) at (0, 2) {};
		\node [style=dark] (25) at (0, -3) {};
		\node [style=white box] (26) at (2.5, -2) {};
		\node [style=white box] (27) at (4.5, 0) {};
		\node [style=none] (28) at (2.5, 2) {};
		\node [style=none] (29) at (6.5, 0) {};
		\node [style=none] (30) at (0, -1) {};
		\node [style=white box] (31) at (3.5, -1.5) {};
		\node [style=white box] (32) at (5.5, 0) {};
		\node [style=white box] (33) at (1, 2) {};
		\node [style=none] (34) at (0, -3) {$\neg$};
		\node [style=none] (35) at (-3, 3) {$x_1$};
		\node [style=none] (36) at (-3, 1) {$x_2$};
		\node [style=none] (37) at (-3, -1) {$x_1$};
		\node [style=none] (38) at (-3, -3) {$x_3$};
		\node [style=none] (39) at (8.5, 0) {$\rightarrow$};
		\node [style=none] (40) at (14.5, 3) {};
		\node [style=none] (41) at (14.5, 1) {};
		\node [style=none] (42) at (14.5, -1) {};
		\node [style=none] (43) at (14.5, -3) {};
		\node [style=white box] (44) at (16.5, 2) {};
		\node [style=dark] (45) at (16.5, -3) {};
		\node [style=white box] (46) at (19, -2) {};
		\node [style=white box] (47) at (21, 0) {};
		\node [style=none] (48) at (19, 2) {};
		\node [style=none] (49) at (23, 0) {};
		\node [style=none] (50) at (16.5, -1) {};
		\node [style=white box] (51) at (20, -1.5) {};
		\node [style=white box] (52) at (22, 0) {};
		\node [style=white box] (53) at (17.5, 2) {};
		\node [style=none] (54) at (16.5, -3) {$\neg$};
		\node [style=white dot] (59) at (12, 0) {};
		\node [style=white dot] (60) at (12, -2) {};
		\node [style=white dot] (61) at (12, 2) {};
		\node [style=none] (62) at (11, 2) {};
		\node [style=none] (63) at (11, 0) {};
		\node [style=none] (64) at (11, -2) {};
		\node [style=halfscalar] (65) at (2.75, -1) {};
		\node [style=halfscalar] (66) at (5, 1) {};
		\node [style=halfscalar] (67) at (0.5, 3) {};
		\node [style=halfscalar] (68) at (19.25, -1) {};
		\node [style=halfscalar] (69) at (21.5, 1) {};
		\node [style=halfscalar] (70) at (17, 3) {};
	\end{pgfonlayer}
	\begin{pgfonlayer}{edgelayer}
		\draw [bend right=15] (10) to (6.center);
		\draw [bend right=15] (7.center) to (10);
		\draw (10)
			 to (15.center)
			 to [bend left] (14);
		\draw (14) to (16.center);
		\draw [bend left] (14) to (13);
		\draw [bend left=15] (13) to (12);
		\draw (12) to (9.center);
		\draw (8.center) to (17.center);
		\draw [bend left=15] (17.center) to (13);
		\draw [bend right=15] (24) to (20.center);
		\draw [bend right=15] (21.center) to (24);
		\draw (24)
			 to (28.center)
			 to [bend left] (27);
		\draw (27) to (29.center);
		\draw [bend left] (27) to (26);
		\draw [bend left=15] (26) to (25);
		\draw (25) to (23.center);
		\draw (22.center) to (30.center);
		\draw [bend left=15] (30.center) to (26);
		\draw (44)
			 to [bend right=15] (40.center)
			 to [bend right=15] (61);
		\draw (41.center) to (44);
		\draw (44) to (48.center);
		\draw [bend left] (48.center) to (47);
		\draw (47) to (49.center);
		\draw [bend left] (47) to (46);
		\draw [bend left=15] (46) to (45);
		\draw (60)
			 to [bend right=15] (43.center)
			 to (45);
		\draw (61)
			 to [bend right=15] (42.center)
			 to (50.center);
		\draw [bend left=15] (50.center) to (46);
		\draw (59) to (41.center);
		\draw (62.center) to (61);
		\draw (59) to (63.center);
		\draw (64.center) to (60);
	\end{pgfonlayer}
\end{tikzpicture}}
    \caption{{{$(x_1 \wedge x_2) \wedge (x_1 \wedge \neg x_3)$}} in phase-free ZH}
    \label{fig:formula example}
\end{figure}
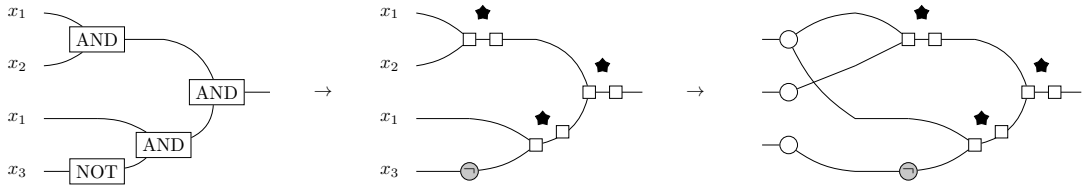

From now on, given a boolean formula $\phi$ on some variables $x_1, \dots, x_n$, its diagram representation is denoted $D_{\phi}$. $D_{\phi}$ has $n$ left dangling edges corresponding to the variables and one right dangling edge, corresponding to the valuation of $\phi$. Attaching $\ket{0}$s or $\ket{1}$s to the left dangling edges, results in the evaluation of the formula for some given assignment. However, we can also efficiently represent the number of satisfying assignments of a formula. By attaching white spiders to edges corresponding to the variables, we evaluate the diagram on the state $\sqrt{2}^n \ket{+\dots+}$, which equals the superposition of all states from computational basis, i.e. we evaluate the formula on all possible assignments simultaneously, as captured in the following lemma.

\begin{lemma}\label{formulaattached}
    Let $\phi$ be a boolean formula on variables $x_1, \dots, x_n$ and $D_{\phi}$ be the phase-free ZH encoding of $\phi$. The diagram obtained by attaching white spiders to the dangling edges corresponding to the variables has matrix representation $\csat(\phi)\ket{1} + (2^n - \csat(\phi))\ket{0}$.
\end{lemma}
\begin{proof}
    We have:
    \begin{align*}
        \diageval[D_{\phi}] (\sqrt{2}^n\ket{+\dots+}) &=
        \sum_{v_1,\dots,v_n \in \{ 0,1\}} \diageval[D_{\phi}] \ket{v_1 \dots v_n} \\&=
        \sum_{v_1,\dots,v_n \in \{0,1\}} \ket{v(\phi)} \\&=
        \csat(\phi)\ket{1} + (2^n - \csat(\phi))\ket{0}
    \end{align*}
    where $v$ stands for the valuation of $x_1, \dots, x_n$ mapping $x_i$ to $True$ when $v_i = 1$ and to $False$ when $v_i = 0$.
\end{proof}

\begin{example}
    Consider the previous formula {$\phi := (x_1 \wedge x_2) \wedge (x_1 \wedge \neg x_3)$}. By the above lemma and the fact that $\csat(\phi) = 1$, we get the following:
    \begin{equation*}
        \diageval[\scalebox{0.77}
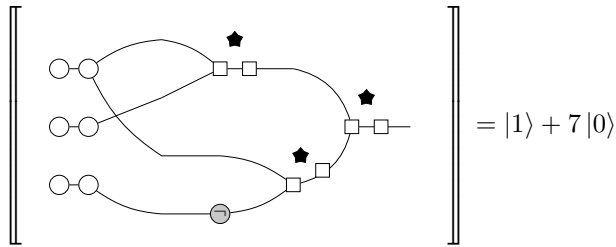
{\begin{tikzpicture}
	\begin{pgfonlayer}{nodelayer}
		\node [style=none] (36) at (-2.5, 3) {};
		\node [style=none] (37) at (-2.5, 1) {};
		\node [style=none] (38) at (-2.5, -1) {};
		\node [style=none] (39) at (-2.5, -3) {};
		\node [style=white box] (40) at (-0.5, 2) {};
		\node [style=dark] (41) at (-0.5, -3) {};
		\node [style=white box] (42) at (2, -2) {};
		\node [style=white box] (43) at (4, 0) {};
		\node [style=none] (44) at (2, 2) {};
		\node [style=none] (45) at (6, 0) {};
		\node [style=none] (46) at (-0.5, -1) {};
		\node [style=white box] (47) at (3, -1.5) {};
		\node [style=white box] (48) at (5, 0) {};
		\node [style=white box] (49) at (0.5, 2) {};
		\node [style=none] (50) at (-0.5, -3) {$\neg$};
		\node [style=white dot] (51) at (-5, 0) {};
		\node [style=white dot] (52) at (-5, -2) {};
		\node [style=white dot] (53) at (-5, 2) {};
		\node [style=white dot] (54) at (-6, 2) {};
		\node [style=white dot] (55) at (-6, 0) {};
		\node [style=white dot] (56) at (-6, -2) {};
		\node [style=halfscalar] (60) at (2.25, -1) {};
		\node [style=halfscalar] (61) at (4.5, 1) {};
		\node [style=halfscalar] (62) at (0, 3) {};
		\node [style=none] (63) at (7, 0) {};
		\node [style=none] (64) at (-7, 0) {};
		\node [style=none] (65) at (0, 4) {};
		\node [style=none] (66) at (0, -4) {};
	\end{pgfonlayer}
	\begin{pgfonlayer}{edgelayer}
		\draw (40)
			 to [bend right=15] (36.center)
			 to [bend right=15] (53);
		\draw (37.center) to (40);
		\draw (40) to (44.center);
		\draw [bend left] (44.center) to (43);
		\draw (43) to (45.center);
		\draw [bend left] (43) to (42);
		\draw [bend left=15] (42) to (41);
		\draw (52)
			 to [bend right=15] (39.center)
			 to (41);
		\draw (53)
			 to [bend right=15] (38.center)
			 to (46.center);
		\draw [bend left=15] (46.center) to (42);
		\draw (51) to (37.center);
		\draw (54) to (53);
		\draw (51) to (55);
		\draw (56) to (52);
	\end{pgfonlayer}
\end{tikzpicture}}] = \ket{1} + 7\ket{0}
    \end{equation*}
\end{example}

Given a boolean formula $\phi$, we can use the above construction to represent $\csat(\phi)$ in phase-free ZH without computing it explicitly. Similar constructions can be found for example in \cite{debeaudrapTensorNetworkRewriting2021, debeaudrapCircuitExtractionZXDiagrams2022c, laakkonenGraphicalSATAlgorithm2022, laakkonenPicturingCountingReductions2023}. The construction also relates to the category approach from \cite{comfortZXCalculusComplete2021}.

\section{\texorpdfstring{$\npsharpp$}{NP\#P}-completeness}\label{npsharppsection}

In this section, we prove theorem \ref{main result}. The proof has four essential parts, each covered in a separate subsection. First, we show that $\StateEq$ and $\ContainsEntryk$ both are in $\npsharpp$. The proof of hardness is more involved and is spread over three subsections. In the second subsection, we show that $\npsharpp = \npcep$. $1$ in brackets means that the oracle is called precisely once during the computation. In the third subsection, we craft an artificial $\npcep$-complete problem. In the final subsection, we reduce such a problem to both $\StateEq$ and $\ContainsEntryk$.

\subsection{Containment}
We need the following problem:

\begin{quote}
    $\ScalarDiagram$\\
    \textbf{Input:} a phase-free ZH diagram $D$ with no dangling edges (i.e. a scalar diagram).\\
    \textbf{Output:} the number $\diageval[D]$.
\end{quote}

It is known that this problem is in $\RM{FP}^{\counting{P}}$ \cite{laakkonenPicturingCountingReductions2023}. A version that works for some tensor networks other than only phase-free ZH follows from the Holant framework \cite[page 212]{caiComplexityDichotomiesCounting2017} and \cite[Lemma 54]{backensFullDichotomyHolant2021}. Yet another explanation was given in \cite{debeaudrapCircuitExtractionZXDiagrams2022c}, where authors used the method introduced in \cite{adlemanQuantumComputability1997} to obtain a similar bound for circuits rather than a diagram, though as the authors point out the method works for diagrams as well. By pushing the deterministic polynomial time operations to instead be performed by the NDTM: $\RM{NP}^{\RM{FP}^{\counting{P}}} \subseteq \npsharpp$. Hence, for the containment of $\StateEq$ and $\ContainsEntryk$ in $\npsharpp$, it suffices to show containment in $\RM{NP}^{\ScalarDiagram}$.

\begin{theorem}\label{cont stateeq}
    The problem $\StateEq$ is in $\npsharpp$.
\end{theorem}
\begin{proof}
    We construct a polytime NDTM $M$ with $\ScalarDiagram$ oracle that takes input $(D_1, D_2)$ and non-deterministically chooses a computational basis state $\ket{v}$. Next, using the oracle, $M$ computes both $\diageval[D_1]\ket{v}$ and $\diageval[D_2]\ket{v}$. Finally, $M$ accepts if the two values are equal.
\end{proof}

For the second problem, we proceed similarly.

\begin{theorem}\label{cont containsentry}
    The problem $\ContainsEntryk$ is in $\npsharpp$.
\end{theorem}
\begin{proof}
    We construct a polytime NDTM $M$ with $\ScalarDiagram$ oracle that takes input $D$ and non-deterministically chooses a computational basis state $\ket{v}$. Next, using the oracle, $M$ computes $\diageval[D]\ket{v}$. Finally, $M$ accepts if the obtained value equals $k$.
\end{proof}

\subsection{Hardness -- oracle change}
To simplify proofs of $\npsharpp$-hardness of $\StateEq$ and $\ContainsEntryk$, we use the following fact.

\begin{theorem}\label{npsharppisnpcep}
    $\npsharpp = \npcep$. Further, the equality holds if the answer to the single call to the $\cep$ oracle is returned as an answer.
\end{theorem}

The idea is as follows: instead of calling $\csat$ oracle, a NDTM can non-deterministically choose the answer. Then, at the end of the computation, it can verify all answer choices with a single call to the $\ComparecSAT$ oracle.

To prove theorem \ref{npsharppisnpcep} formally, we need a few additional lemmas.

\begin{lemma}[Concat formulae]\label{concat formulae}
    Given two boolean formulae $\phi$ and $\psi$, it is possible to deterministically construct a boolean formula $\rho$, such that:
    \begin{itemize}
        \item $\rho$ can be constructed in time polynomial in the sizes of $\phi$ and $\psi$,
        \item $\csat(\rho)$ written in binary is a concatenation of $\csat(\phi)$ and $\csat(\psi)$ written in binary, possibly with some additional leading zeros.
    \end{itemize}
\end{lemma}

\begin{proof}
    Let $x_1,\dots,x_n$ be variables of $\phi$ and $y_1,\dots,y_m$ be variables of $\psi$. By substituting variables in one formula with new ones, we can assume that $x_i \ne y_j$ for all $i,j$. We define $\rho$ as a boolean formula on $x_1,\dots,x_n,y_1,\dots,y_m,z,z'$ as follows:
    \begin{equation*}
        \rho = (\phi \wedge z) \vee (x_1 \wedge \dots \wedge x_n \wedge \psi \wedge \neg z \wedge z')
    \end{equation*}
    Let:
    \begin{align*}
        X &= x_1,\dots,x_n\\
        Y &= y_1,\dots,y_m\\
        Z &= z,z'\\
        A &= X,Y,Z
    \end{align*}
    Then, we have:
    \begin{align*}
        \csat(\rho, A) &= \csat(\phi \wedge z, A) + \csat(x_1 \wedge \dots \wedge x_n \wedge \psi \wedge \neg z \wedge z', A) \\&= \csat(\phi,X)\cdot 2^{m+1} + \csat(\psi,Y)
    \end{align*}
    The first equality holds as $\rho$ is a conjunction of two formulae that cannot be simultaneously satisfied due to the requirement of different valuations of $z$. The second equality follows from the fact, that $\phi \wedge z$ is independent of valuations of $Y,z'$. Thus, the last $m+1$ digits of $\csat(\rho,A)$ in binary represent $\csat(\psi,Y)$ possibly with leading zeros, and the other first digits of $\csat(\rho,A)$ represent $\csat(\phi,X)$.
\end{proof}

\begin{lemma}[Concat $\ComparecSAT$ instances]\label{concat ccsat}
    Let $(\phi, \psi)$ and $(\zeta, \xi)$ be two instances for problem $\ComparecSAT$. They can be combined to a single instance for which the answer is $True$ if and only if the answers are $True$ for the two given instances.
\end{lemma}

\begin{proof}
    Let $\rho$ be a formula constructed from $\phi$ and $\zeta$ in the previous lemma and $\tau$ be a formula constructed from $\psi$ and $\xi$. Then the instance $(\rho,\tau)$ can be constructed in polytime and by construction: $\csat(\rho) = \csat(\phi) \cdot 2^{m+1} + \csat(\zeta)$ and $\csat(\tau) = \csat(\psi) \cdot 2^{m+1} + \csat(\xi)$, where $m$ is the number of variables of $\zeta$ (and, thus, also of $\xi$). Since $\csat(\zeta)$ and $\csat(\xi)$ are both bounded above by $2^m$, we have:
    \begin{equation*}
        \csat(\rho) = \csat(\tau) \Leftrightarrow (\csat(\xi) = \csat(\zeta)) \wedge (\csat(\phi) = \csat(\psi))
    \end{equation*}
    as required.
\end{proof}

Another necessary construction is the creation of a boolean formula with a given number of satisfying assignments.

\begin{lemma}\label{number to formula}
    Let $x_1,\dots,x_n$ be variables and $k$ a number in $[0..2^n]$. Then, it is possible to construct a boolean formula $\psi$ on $x_1,\dots,x_n$ with $\csat(\psi, (x_1,\dots,x_n)) = k$.
\end{lemma}

\begin{proof}[Proof of lemma \ref{number to formula}]
    Let $a_0, \dots, a_{n}$ be the digits of $k$ written in binary, possibly with leading zeros. Thus, we have $k = \sum_{i \in [0..n]} 2^{n-i} \cdot a_i = \overline{(a_0 \dots a_n)}_2$. Now, define the following formula:
    \begin{align*}
        \psi := \neg l(a_0) \to (x_1 \wedgevee_{a_1} (x_2 \wedgevee_{a_2} \dots (x_{n-1} \wedgevee_{a_{n-1}} (x_n \wedge l(a_n)))\dots))
    \end{align*}
    where $\wedgevee_0 = \wedge$ and $\wedgevee_1 = \vee$, and $l(0) = False$ and $l(1) = True$.
    We prove that $\csat(\rho) = k$.

    When $n = 0$, then the formula simplifies to $\neg l(a_0) \to l(a_0)$, which is equivalent to $True \to False = False$ for $a_0 = 0$ and $False \to True = True$ for $a_0 = 1$. So $\csat(\psi) = a_0 = k$.

    When $a_0 = 1$, then $k \ge 2^n$. Since $k \in [0..2^n]$, we must have $k = 2^n$. On the other hand, when $a_0 = 1$, then $\rho$ simplifies to a tautology $False \to (\dots)$, and thus $\csat(\psi) = 2^n = k$, as required.

    The remaining cases have $n \ge 1$ and $a_0 = 0$ (so $k < 2^n$). Thus, we can ignore the ``$\neg l(a_0) \to$'' part of the formula. For these cases, we proceed by induction on $n$.

    When $n = 1$, then $\psi$ simplifies to $x_1 \wedge l(a_1)$, which has $a_1 = \overline{(0a_1)}_2 = k$ satisfiable assignments.

    For the induction step, we assume that the thesis holds up to some $n-1$, and we prove it also holds for $n$. Let $\tau = x_2 \wedgevee_{a_2} (x_3 \dots \wedgevee_{a_{n-1}} (x_k \wedge l(a_n))\dots )$ be the part of the formula after $x_1 \wedgevee_{a_1}$. Thus, by induction hypothesis $\csat(\tau, (x_2,\dots,x_n)) = \overline{(a_2 \dots a_n)}_2$.
    
    When $a_1 = 0$, then $\psi = x_1 \wedge \tau$ and $\csat(\psi, (x_1,\dots,x_n)) = \csat(\tau, (x_2,\dots,x_n)) = \overline{(a_2 \dots a_n)}_2 = \overline{(0a_2\dots a_n)}_2 = k$.
    
    When $a_1 = 1$, then $\psi = x_1 \vee \tau$. Since $x_1$ does not appear in $\tau$, we have $\csat(\psi, (x_1,\dots,x_n)) = 2^{n-1} + \csat(\tau, (x_2,\dots,x_n)) = 2^{n-1} + \overline{(a_2\dots a_n)}_2 = \overline{(1a_2\dots a_n)}_2 = k$, which ends the proof.
\end{proof}

\begin{example}
All possible formulae when using $3$ bits for the number of satisfying assignments (i.e. integers from $0$ to $4$):
\begin{center}
\begin{tabular}{| c | c c c | c | c | c |}
\hline
$k$ & $a_0$ & $a_1$ & $a_2$ & $\neg l(a_0) \to (x_1 \wedgevee_{a_1} (x_2 \wedge l(a_2)))$ & Simplified formula & $x_1, x_2$ assignments \\
\hline
$0$ & $0$ & $0$ & $0$ & $True \to (x_1 \wedge x_2 \wedge False)$ & $False$ & None \\
$1$ & $0$ & $0$ & $1$ & $True \to (x_1 \wedge x_2 \wedge True)$ & $x_1 \wedge x_2$ & $11$ \\
$2$ & $0$ & $1$ & $0$ & $True \to (x_1 \vee (x_2 \wedge False))$ & $x_1$ & $10, 11$ \\
$3$ & $0$ & $1$ & $1$ & $True \to (x_1 \vee (x_2 \wedge True))$ & $x_1 \vee x_2$ & $01, 10, 11$ \\
$4$ & $1$ & $0$ & $0$ & $False \to (x_1 \wedge x_2 \wedge False)$ & $True$ & $00, 01, 10, 11$ \\
\hline
\end{tabular}
\end{center}
\end{example}

Finally, we can prove theorem \ref{npsharppisnpcep}.

\begin{proof}[Proof of theorem \ref{npsharppisnpcep}]
    We start by proving the right containment $\npsharpp \subseteq \npcep$.

    $\cSAT$ is complete for class $\counting{P}$, thus $\npsharpp = \RM{NP}^{\cSAT}$. Let $A$ be a problem in $\RM{NP}^{\cSAT}$. It suffices to show that $A \in \npcep$. Let $M$ be a polytime NDTM with access to the $\cSAT$ oracle that recognizes $A$. Define NDTM $M'$ as follows. $M'$ contains one extra tape over $M$ and for all inputs $w$, the transition function of $M'$ matches that of $M$ (ignoring the extra tape of $M'$) for all configurations except for the three states: the oracle ask $q_{ASK}$, the accepting state $q_{ACC}$ and the rejecting state $q_{REJ}$. Further, $M'$ has access $\ComparecSAT$ oracle rather than $\cSAT$ oracle.

    When $M$ enters the oracle query state $q_{ASK}$ and sends $\phi, (x_1,\dots,x_n)$ to the $\cSAT$ oracle, then $M'$ instead non-deterministically chooses an answer $k\in[0..2^n]$ to the oracle query. Next, it creates a boolean formula $\psi$ on $x_1,\dots,x_n$ such that $\csat(\psi) = k$ by using the lemma \ref{number to formula}. Then, $M'$ writes the $\ComparecSAT$ instance $(\phi,\psi)$ on its extra tape. Finally, $M$ moves to the configuration that $M$ would be in if it received the answer $k$ from the oracle (except for the extra tape).

    When $M$ enters the accepting state $q_{ASK}$, then $M'$ instead constructs two boolean formulae $(\Psi,\Phi)$ such that $\csat(\Psi) = \csat(\Phi)$ if and only if $\csat(\phi) = \csat(\psi)$ for all pairs $(\phi,\psi)$ stored in the extra tape of $M'$. Finally, $M'$ sends $(\Phi,\Psi)$ to its oracle and returns the oracle's answer. Note, that the construction of $(\Phi,\Psi)$ is always possible via a series of concatenations presented in lemma \ref{concat ccsat}.

    When $M$ enters the rejecting state $q_{REJ}$, $M'$ calls its oracle with a $False$ instance and returns its answer.

    The NDTM $M'$ runs in a polytime since $M$ does, and all extra operations of $M'$ can be performed in time polynomial in the input size. Further, $M$ recognizes the same language. To see this, consider $w \in A = \RM{L}(M)$. Then, $M$ accepts $w$ in one of its nondeterministic paths $p$. Then in one of its nondeterministic paths, $M'$ picks the correct answer to the first oracle query that $M$ asks for on the path $p$. In that path, $M'$ further can non-deterministically choose a path that also picks the correct answer to the second oracle query $M$ asks for on the path $p$, then the third, etc. In the end, $M$ enters $q_{ACC}$, so $M'$ constructs a single instance of $\ComparecSAT$ that is used to verify that all nondeterministically chosen answers to the $M$ oracle queries were correct, so, in the end, $M'$ accepts $w$. When $M$ does not accept $w$, then on all its paths $M$ eventually enters the rejecting state. Thus $M'$ either ends with an unsatisfiable instance for the query (when $M'$ enters equivalent of $q_{REJ}$ state of $M$), or it had to choose an incorrect answer for at least one of the queries, which is captured in the end with the $\ComparecSAT$ oracle.

    The left containment $\npsharpp \supseteq \npcep$ is immediate. Instead of sending $\ComparecSAT$ query $(\phi, \psi)$, a Turing Machine can ask for $\csat(\phi)$ and $\csat(\psi)$, and then compare the answers to effectively simulate $\ComparecSAT$ oracle with two calls to the $\cSAT$ oracle.
\end{proof}

We could not find a shorter proof of the above fact. However, it is worth mentioning how a shorter proof could proceed. By corollary of Toda's theorem \cite{todaPPHardPolynomialTime1991}: $\RM{P}^{\counting{P}} = \RM{P}^{\RM{PP}}$ and hence $\npsharpp = \RM{NP}^{\RM{PP}}$. Then it may be possible to adjust Tor\'an's work \cite[Corollary 3.13 and Theorem 4.1]{toranComplexityClassesDefined1991} to obtain $\RM{NP}^{\RM{PP}} = \RM{NP}^{\RM{C_=P}}$. Finally, $\RM{NP}^{\RM{C_=P}} = \npcep$ by, for instance, paring functions. We would skip the Turing Machines here, however, they are necessary in the remaining parts of the proof anyway.

Thanks to the above theorem, instead of showing $\npsharpp$-hardness, we can show $\npcep$-hardness.

\subsection{Hardness -- crafting complete problem}

We want to show $\npsharpp$-hardness. By theorem \ref{npsharppisnpcep}, it is equivalent to showing $\npcep$-hardness. There are no natural $\npcep$-complete problems suitable for a reduction, so we craft an artificial $\npcep$-complete problem $\SATandComparecSAT$.

\begin{quote}
    $\SATandComparecSAT$\\
    \textbf{Input:} Natural numbers $n, m$ and two boolean formulae:\begin{itemize}
        \item $\psi$, defined on $x_1, \dots, x_n, y_1, \dots, y_m$,
        \item $\rho$, defined on $x_1, \dots, x_n, z_1, \dots, z_m$
    \end{itemize}
    \textbf{Output:} $True$ if and only if there exists valuation $v \colon \{ x_1, \dots, x_n \} \to \{ True, False \}$ such that:
    \begin{equation*}
        \csat(v(\psi), y_1,\dots,y_m) = \csat(v(\rho), z_1,\dots,z_m)
    \end{equation*}
\end{quote}

Note, that in the problem definition, $v(\psi)$ and $v(\rho)$ are formulae on variables $y_1, \dots, y_m$ and $z_1, \dots, z_m$ respectively.

While the definition of the problem is lengthy, we could argue that it is the most natural problem for $\npcep$ -- it combines $\SAT$ and $\ComparecSAT$ so the natural problems for $\RM{NP}$ and $\RM{C_=P}$.

The other existing complete problems for $\npsharpp$ \cite{toranComplexityClassesDefined1991, monniauxNPExistsPP2022} are based on checking which answer to two $\cSAT$ instances is greater, or how to maximize an answer, rather than asking for exact equality. Comparison of which value is greater likely cannot be directly represented in phase-free ZH.

\begin{example}\label{satccsatexample}
    Let $n=1, m=2$ and:\begin{itemize}
        \item $\psi = x_1 \vee y_1 \vee \neg y_2$,
        \item $\rho = \neg (z_1 \wedge (z_2 \vee x_1))$.
    \end{itemize}
    Then, under valuation $v$ mapping $x_1$ to $False$, we have:
    \begin{align*}
        v(\psi) &= False \vee y_1 \vee \neg y_2 = y_1 \vee \neg y_2\\
        v(\rho) &= \neg(z_1 \wedge (z_2 \vee False)) = \neg(z_1 \wedge z_2) = \neg z_1 \vee \neg z_2
    \end{align*}
    Since:
    \begin{gather*}
        \csat(v(\psi)) = \csat(y_1 \vee \neg y_2) = 3 = \csat(\neg z_1 \vee \neg z_2) = \csat(v(\rho))
    \end{gather*}
    the answer to such an instance of $\SATandComparecSAT$ is $True$. For valuation $v'$ mapping $x_1$ to $True$, the two values do not equalize:
    \begin{align*}
        v'(\psi) &= True \vee y_1 \vee \neg y_2 = True\\
        v'(\rho) &= \neg(z_1 \wedge (z_2 \vee True)) = \neg(z_1 \wedge True) = \neg z_1
    \end{align*}
    and:
    \begin{align*}
        \csat(v'(\phi), y_1,y_2) = \csat(True, y_1,y_2) = 4 \ne 2 = \csat((\neg z_1), z_1,z_2) = \csat(v'(\rho), z_1,z_2).
    \end{align*}
\end{example}

The rest of this subsection is dedicated to showing the $\npcep$-completeness of the above problem. The proof involves typical transformations of Turing machines. The next part of the proof is in the subsection \ref{zhencodingsection}.

We put a plethora of constraints on TMs for problems in $\npcep$. One of these constraints describes how a TM should represent a boolean formula used in the query. We want TMs to work solely with $0$ and $1$ symbols. The exact representation of the boolean formula in binary is not essential -- we only require it to exist and it must be possible to deterministically transform a sequence of zeros and ones back into a formula.

\begin{definition}[CNF]
    \sloppy A boolean formula on $x_1, \dots, x_n$ is in \textbf{conjunctive normal form} (\textbf{CNF}) if it is a conjunction of clauses, where each clause is a disjunction of some of the literals $x_1, \neg x_1, \dots, x_n, \neg x_n$.
\end{definition}

\begin{definition}[$0-1$ encoding of CNFs]
    Let $\phi = c_1 \wedge c_2 \wedge \dots \wedge c_m$ be a boolean formula in CNF on variables $x_1, \dots, x_n$. We define $0-1$ encoding of $\phi$ as a word $e_{\phi}$, defined as follows, where $k = \max(m,n)$:
    \begin{itemize}
        \item Extend set of allowed variables by $k-n$ variables $x_{n+1}, \dots, x_k$,
        \item Extend $\phi$ with $k-m$ clauses containing literals $x_1$ and $\neg x_1$,
        \item Extend each clause with literals $x_{n+1}, \dots, x_k$
        \item Translate each clause $c_j$ of $\phi$ to $2k$ symbols from $\{ 0,1 \}$, where symbols at the positions $2i-1$ and $2i$ correspond to variable $x_i$ for $i \in [1..k]$ as follows:
        \begin{itemize}
            \item The first symbol is $0$ when $c_j$ does not contain literal $x_i$ and $1$ otherwise,
            \item The second symbol is $0$ when $c_j$ does not contain literal $\neg x_i$ and $1$ otherwise.
        \end{itemize}
    \end{itemize}
\end{definition}

\begin{example}
    The CNF formula $x_1 \wedge (x_2 \vee \neg x_3)$ is translated to the following word $e_{x_1 \wedge (x_2 \vee \neg x_3)}$ (spaces are added for visibility only):
    \begin{equation*}
        10\ 00\ 00\ \ 00\ 10\ 01\ \ 11\ 00\ 00
    \end{equation*}
\end{example}

\begin{definition}
    Let $w$ be any word over $\{0,1\}$ with $2k^2$ symbols for some $k$. We define formula of $w$ as $\textsc{Formula}(w)$ defined as follows:
    \begin{itemize}
        \item $\textsc{Formula}(w)$ has $k$ clauses and is defined on $k$ variables $x_1, \dots, x_k$,
        \item Each clause corresponds to a block of $2k$ symbols, where the next $2$ symbols correspond to literals describing variable $x_i$, same as in the definition of $0-1$ encoding:
        \begin{itemize}
            \item If the first symbol is $1$, add literal $x_i$, otherwise add literal $False$,
            \item If the second symbol is $1$, add literal $\neg x_i$, otherwise add literal $False$.
        \end{itemize}
    \end{itemize}
\end{definition}

\begin{example}
    Consider the previous word:
    \begin{equation*}
        10\ 00\ 00\ \ 00\ 10\ 01\ \ 11\ 00\ 00
    \end{equation*}
    It corresponds to the formula below, where $0$ is used as shorthand for $False$:
    \begin{equation*}
        (x_1 \vee 0 \vee 0 \vee 0 \vee 0 \vee 0) \wedge (x_2 \vee 0 \vee 0 \vee 0 \vee 0 \vee \neg x_3) \wedge (x_1 \vee \neg x_1 \vee 0 \vee 0 \vee 0 \vee 0)
    \end{equation*}
    By removing $False$ literals from clauses, we get:
    \begin{equation*}
        x_1 \wedge (x_2 \vee \neg x_3) \wedge (x_1 \vee \neg x_1)
    \end{equation*}
    By removing the last clause which is satisfied by any assignment, we get the original formula:
    \begin{equation*}
        x_1 \wedge (x_2 \vee \neg x_3)
    \end{equation*}
\end{example}

\begin{theorem}\label{01encworks}
    Let $\Phi$ be a boolean formula in CNF with $m$ clauses, and let $x_1,\dots,x_n$ be variables on which it is considered. Then:
    \begin{equation*}
        \csat(\Phi, (x_1,\dots,x_n)) = \csat(\textsc{Formula}(e_{\Phi}), (x_1,\dots,x_{\max(n,m)}))
    \end{equation*}
\end{theorem}
\begin{proof}
    Let $k = \max(m,n)$. Consider steps in the construction of $e_{\Phi}$:
    \begin{itemize}
        \item Extension of $\Phi$ by the clauses containing literals $x_1$ and $\neg x_1$ does not change satisfiable assignments -- any assignment satisfies such clauses,
        \item \sloppy Extension of set of variables to $x_1, \dots, x_k$, while extending each clause with literals $x_{n+1},\dots,x_k$ does not change number of satisfiable assignments -- the only satisfiable assignments after the extension are those with $x_{n+1},\dots,x_k$ set to $True$ and $x_1,\dots,x_n$ set to satisfiable assignment of the original $\Phi$.
    \end{itemize}
    Let $\Psi$ be the formula $\Phi$ after these steps, therefore:
    \begin{equation*}
        \csat(\Phi, (x_1,\dots,x_n)) = \csat(\Psi, (x_1,\dots,x_k))
    \end{equation*}
    The encoding of literals preserves the satisfiable assignments when the number of clauses equals the number of variables. Therefore: $\textsc{Formula}(e_{\Psi}) = \Psi$. By construction $e_{\Psi} = e_{\Phi}$ and thus the thesis follows:
    \begin{align*}
        \csat(\textsc{Formula}(e_{\Phi}), (x_1,\dots,x_k))
        &=\csat(\textsc{Formula}(e_{\Psi}), (x_1,\dots,x_k)) \\
        &=\csat(\Psi, (x_1,\dots,x_k)) \\
        &=\csat(\Phi, (x_1,\dots,x_n))
    \end{align*}
\end{proof}

We can now put constraints on a TM for some $A \in \npcep$.

\begin{theorem}\label{tmbs}
    Let $A \in \npcep$. Then, there exists a polytime NDTM $\mathcal{M}$ with access to $\ComparecSAT$ oracle, that recognized $A$ and satisfies the following conditions:
    \begin{enumerate}
        \item $\mathcal{M}$ calls the oracle $A$ precisely once, at the very end of the computation, and $\mathcal{M}$ returns the oracle answer,
        \item the boolean formulae that $\mathcal{M}$ sends to the oracle are all in CNF,
        \item $\mathcal{M}$ has two dedicated tapes used for communication with the oracle. When $\mathcal{M}$ enters its $q_{ASK}$ state to ask an $\ComparecSAT$ query $(\zeta,\xi)$, then on the first tape it contains a representation of the first boolean formula $\zeta$ as a sequence of zeros and ones $e_{\zeta}$, and on the second tape it contains a representation $e_{\xi}$ of $\xi$,
        \item there exists a polynomial $p$ such that for any input $w$, when $\mathcal{M}$ calls the oracle with $(\Phi,\Psi)$ query then $|e_{\Phi}| = |e_{\Psi}| = p(|w|)$. In other words, the size of the query depends only on the size of the input, not on the input itself.
        \item there exist a polynomial $q$ such that for any input $w$, $\mathcal{M}$ runs in precisely $q(|w|)$ steps.
    \end{enumerate}
\end{theorem}
\begin{proof}
    By theorem \ref{npsharppisnpcep}, there exists polytime NDTM $M$ with access to $\ComparecSAT$ oracle satisfying the first condition.
    
    Since every boolean formula can be transformed into CNF in polynomial time, we can assume that $M$ also satisfies the second condition.
    
    The third condition is just a formal description of how the machine should communicate with the oracle. By theorem \ref{01encworks}, the encoding of the formula does not alter the number of satisfying assignments, and thus we can assume that $M$ further satisfies the third condition.

    The fourth condition is more involved. Let $p_M$ be the polynomial bounding number of steps of $M$, i.e. for input $w$, $M$ terminates in at most $p_M(|w|)$ steps. Define $M'$ as a TM that follows $M$, however, when $M$ enters $q_{ASK}$ state, $M'$ instead changes the contents of oracle tapes: the boolean formulae are transformed to have encoding sizes $2p_M(|w|)^2$ each -- via the same procedure as in the construction of the encoding, we add clauses and variables to have $p_M(|w|)$ of each.
    
    Therefore, $M'$ satisfies the first three conditions, as $M$ does, and it satisfies the fourth condition with $p = 2p_M^2$. Further, $M'$ is polytime, as $M$ is, and extra steps of $M'$ take at most $O(p_M(|w|))$ steps. Let $p_{M'}$ be the polynomial bounding number of steps of $M'$, i.e. for input $w$, $M'$ terminates in at most $p_{M'}(|w|)$ steps.

    Finally, we define $\mathcal{M}$ as a TM that follows $M'$. However, if $M'$ enters $q_{ASK}$ state, then $\mathcal{M}$ stalls for some number of steps so that it always calls the oracle in the $q(|w|)$ step, where $q$ is some polynomial. This can be achieved, for example, by equipping $\mathcal{M}$ with an extra `counter' tape on which it initially fills $p_{M'}(|w|)$ cells with $1$. Then, $\mathcal{M}$ starts computation, following $M'$ and removing $1$ cell from the counter during each step. When the oracle should be called, $\mathcal{M'}$ first continues to clear the counter and only then calls the oracle. Note that $p_{M'}(|w|)$ can be computed based on the input size but not the input itself. Thus, the counter creation and clearance require the number of steps depending only on $|w|$ but not $w$. Therefore, the total number of steps on input $w$ is indeed given by some $q(|w|)$.
\end{proof}

We can now prove that the crafted problem is $\npcep$-complete.

\begin{theorem}\label{nice repr}
    $\SATandComparecSAT$ is $\npcep$-complete.
\end{theorem}
\begin{proof}$ $\\
    \sloppy \indent\textbf{Containment:} given an instance $(n,m,\psi,\rho)$ of $\SATandComparecSAT$, a NDTM with $\ComparecSAT$ oracle can do the following:\begin{itemize}
        \item Non-deterministically choose a valuation of $x_1,\dots,x_n$,
        \item Call the oracle on $(v(\psi), v(\rho))$ and returns its answer.
    \end{itemize}

    \textbf{Hardness:} Let $A \in \npcep$. We must show that problem $A$ can be reduced to $\SATandComparecSAT$. Let $w$ be any instance of $A$. Let $\mathcal{M}, p, q$ be a TM and polynomial guaranteed to exist by theorem \ref{tmbs}.
    
    Let $\Phi_{\mathcal{M},w}$ be a boolean formula constructed by Cook-Levin reduction that describes whether $M$ reaches $q_{ASK}$, and let $X := x_1, x_2, \dots, x_n$ be the variables from Cook-Levin construction. These include variables of the form $y_{t,i,s,k}$ corresponding to statement ``on $k$\textsuperscript{th} step of the computation the $i$\textsuperscript{th} cell of $t$\textsuperscript{th} tape contains symbol $s$'' and similar variables describing the state of TM in the given step of computation and the positions of heads.

    Let $o1$ and $o2$ be the two oracle tapes used to communicate with the oracle, i.e. $o1$ stores the first boolean formula as a sequence, and $o2$ stores the second boolean formula, both as sequences of zeros and ones. By construction, $w \in A$ if and only if in the answer to the final (and only) oracle call is $True$. This oracle call is entirely encoded on the two tapes and thus can be expressed by the same variables as $\phi_{A,w}$.

    Let $v$ be a valuation of the variables $X$ satisfying $\phi_{A,w}$, so a description of the path taken to reach $q_{ASK}$. Let $\Tilde{v} \colon X \to \{0,1\}$ with $\Tilde{v}(x_i) = 1$ when $v(x_i) = True$ and $\Tilde{v}(x_i) = 0$ when $v(x_i) = False$.
    The first stored boolean formula is:
    \begin{equation*}
        \zeta_{\mathcal{M},w,v} = \textsc{Formula}(\Tilde{v}(y_{o1, 1, 1, r(|w|)-1}) \Tilde{v}(y_{o1, 2, 1, r(|w|)-1}) \dots \Tilde{v}(y_{o1, p(|w|), 1, r(|w|)-1}))
    \end{equation*}
    By theorem \ref{tmbs}, only the above variables matter -- the oracle call is always precisely on positions $1$ to $p(|w|)$ of $o1$ tape in the $r(|w|)-1$\textsuperscript{th} step of computation (as the TM in total has $r(|w|)$ steps where the last step is the oracle call). By theorem \ref{01encworks}, the only two possible symbols are $0$ and $1$, and $\Tilde{v}(x_{o1, i, 1, r(|w|)-1})$ is $1$ when the corresponding symbol is $1$ and $0$ when the corresponding symbol is $0$. Note that $p(|w|) = 2k^2$ for some $k \in \mathbb{Z}$, as the $(0-1)$ representation of CNF formulae must be of such length. Let $u_1, \dots, u_k$ be the variables of $\zeta_{\mathcal{M},w,v}$. Let $y'_{l}$ be a shorthand for $y_{o1, l, 1, r(|w|)-1}$. By construction of $\textsc{Formula}$, the value $\Tilde{v}(y'_{2k(i-1)+(2j-1)})$ for $i, j \in [1..k]$ describes whether in the $i$\textsuperscript{th} clause we put literal $u_j$ or $False$. Thus, such literal is equal to the expression $u_j \wedge y'_{2k(i-1)+(2j-1)}$ under the valuation $v$. Similarly, the value $\Tilde{v}(y'_{2k(i-1)+2j})$ describes whether in the $i$\textsuperscript{th} clause we put literal $\neg u_j$ or $False$. Thus, such literal is equal to the expression $\neg u_j \wedge y'_{2k(i-1)+2j}$ under the valuation $v$. Now, define $\Psi_{\mathcal{M},w}$ as follows:
    \begin{gather*}
        \Psi_{\mathcal{M},w} := \psi_1 \wedge \dots \wedge \psi_k,\quad \text{where:}\\
        \psi_i := ((u_1 \wedge y'_{2k(i-1)+1}) \vee (\neg u_1 \wedge y'_{2k(i-1)+2}) \vee \dots \vee (u_k \wedge y'_{2ki-1}) \vee (\neg u_k \wedge y'_{2ki}))
    \end{gather*}
    Then, $v(\psi_i)$ is equal to the $i$\textsuperscript{th} clause of $\zeta_{\mathcal{M},w,v}$, and hence $v(\Psi_{\mathcal{M},w})$ is equal $\zeta_{\mathcal{M},w,v}$.
    
    Similarly, the second stored formula is:
    \begin{equation*}
        \xi_{\mathcal{M},w,v} = \textsc{Formula}(\Tilde{v}(x_{o2, 1, 1, r(|w|)-1}) \Tilde{v}(x_{o2, 2, 1, r(|w|)-1}) \dots \Tilde{v}(x_{o2, p(|w|), 1, r(|w|)-1}))
    \end{equation*}
    Let $u'_1, \dots, u'_k$ be the variables of it. Via a similar procedure, we can produce $\Rho_{\mathcal{M},w}$ such that $v(\Rho_{\mathcal{M},w}) = \xi_{\mathcal{M},w,v}$.

    \sloppy By construction, $w \in A$ if and only if there exists a valuation $v$, such that $\phi_{\mathcal{M},w}$ is satisfied and $\csat(\zeta_{\mathcal{M},w,v}) = \csat(\xi_{\mathcal{M},w,v})$. Equivalently, $v(\Phi_{\mathcal{M},w}) = True$ and $\csat(v(\Psi_{\mathcal{M},w})) = \csat(v(\Rho_{\mathcal{M},w}))$.

    We can now define an equivalent instance of $\SATandComparecSAT$:\begin{itemize}
        \item Numbers $n, k+1$, where $n$ is the number of variables of $\Phi_{\mathcal{M},w}$ and $k$ is the number of variables of $v(\Phi_{\mathcal{M},w})$ and $v(\Rho_{\mathcal{M},w})$,
        \item Boolean formula $\Psi$ defined as $\Phi_{\mathcal{M},w} \wedge (\Psi_{\mathcal{M},w} \vee y)$ on $x_1,\dots,x_n,u_1,\dots,u_k,y$,
        \item Boolean formula $\Rho$ defined as $\Rho_{\mathcal{M},w} \vee y'$ on $x_1,\dots,x_n,u'_1\dots,u'_k,y'$.
    \end{itemize}
    Then $\csat(v(\Psi)) = \csat(v(\Rho))$ under valuation $v$ of $x_1,\dots,x_n$ happens if and only if the following hold, where $U=u_1,\dots,u_k$ and $U' = u'_1,\dots,u'_k$:
    \begin{gather*}
        \csat(v(\Phi_{\mathcal{M},w} \wedge (\Psi_{\mathcal{M},w} \vee y)), U,y) = \csat(v(\Rho_{\mathcal{M},w} \vee y'), U',y)
    \end{gather*}
    Since $v(\Phi_{\mathcal{M},w})$ is either $True$ or $False$, the number $\csat(v(\Phi_{\mathcal{M},w}), \emptyset)$ equals $0$ or $1$. $y$ appears in the formula once, used in disjunction with $\Psi_{\mathcal{M},w}$. Hence we get the following:
    \begin{align*}
        \csat(v(\Phi_{\mathcal{M},w} &\wedge (\Psi_{\mathcal{M},w} \vee y)), U,y) =
        \csat(v(\Phi_{\mathcal{M},w}), \emptyset) \cdot \csat(v(\Psi_{\mathcal{M},w} \vee y), U,y) \\&=
        \csat(v(\Phi_{\mathcal{M},w}), \emptyset) \cdot (\csat(v(\Psi_{\mathcal{M},w}), U) + 2^k) \\&=
        \begin{cases}
            0,\quad &\text{when $\Phi_{\mathcal{M},w}$ is not satisfied under $v$}\\
            \csat(v(\Psi_{\mathcal{M},w}, U)) + 2^k,\quad &\text{when $\Phi_{\mathcal{M},w}$ is satisfied under $v$}
        \end{cases}
    \end{align*}
    Similarly:
    \begin{gather*}
        \csat(v(\Rho_{\mathcal{M},w} \vee y')) = \csat(v(\Psi_{\mathcal{M},w}, U)) + 2^k
    \end{gather*}
    
    \sloppy Since $2^k > 0$, the two expressions equalize precisely when $\Phi_{\mathcal{M},w}$ is satisfied under $v$ and ${\csat(v(\Psi_{\mathcal{M},w}, U)) = \csat(v(\Rho_{\mathcal{M},w}, U))}$. By above, this is equivalent to $\mathcal{M}$ accepting $w$. The above $\SATandComparecSAT$ instance is constructible in polytime in $|w|$, which ends the reduction and the entire proof.
\end{proof}

Now that we have crafted the $\npcep$-complete problem, we finally have a problem to reduce from to finish the main proof.

\subsection{Hardness -- ZH encoding}\label{zhencodingsection}
The goal is to show that $\StateEq$ and $\ContainsEntryk$ are both $\npcep$-hard. It suffices to reduce from $\SATandComparecSAT$. We already know how to represent both a boolean formula and the number of satisfying assignments of a boolean formula in phase-free ZH. Hence, the ZH representation is almost immediate.

\begin{theorem}\label{state equality}
    $\StateEq$ is $\npcep$-hard.
\end{theorem}
\begin{proof}
    We reduce from $\SATandComparecSAT$. Let $(m,n,\psi,\rho)$ be any $\SATandComparecSAT$ instance, where the common variables are $x_1,\dots,x_n$ and extra variables are $y_1,\dots,y_m$ for $\psi$ and $z_1,\dots,z_m$ for $\rho$. Now, define diagrams $D_1, D_2$ as follows:

    
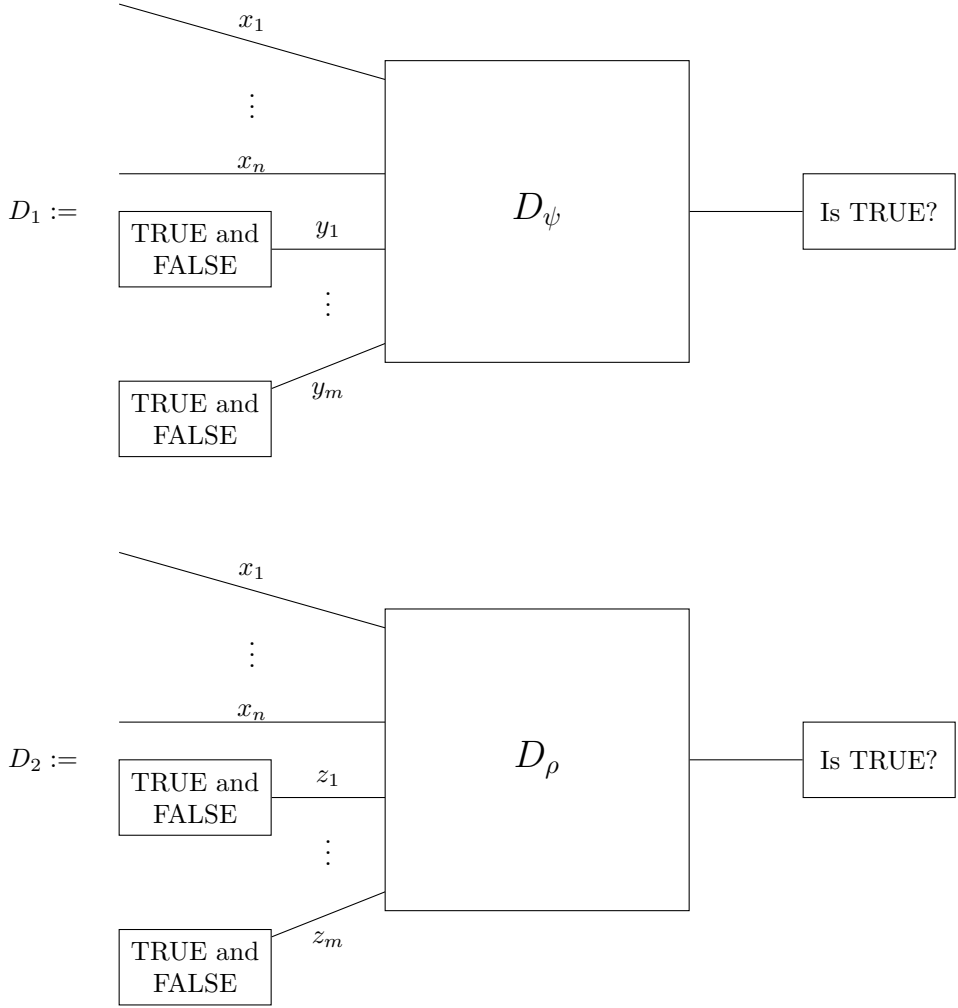
\begin{center}
        \begin{tikzpicture}
	\begin{pgfonlayer}{nodelayer}
		\node [style=white box, minimum width=4cm, minimum height=4cm] (0) at (0, 0) {};
		\node [style=none] (2) at (-11, 5.5) {};
		\node [style=none] (4) at (-11, 1) {};
		\node [style=none] (5) at (-4, 3.5) {};
		\node [style=none] (7) at (-4, 1) {};
		\node [style=none] (8) at (-5.5, -2.25) {$\vdots$};
		\node [style=none] (19) at (-7.5, 5) {$x_1$};
		\node [style=none] (21) at (-7.5, 1.25) {$x_n$};
		\node [style=multiline text] (23) at (0, 0) {\Large$D_{\psi}$};
		\node [style=none] (24) at (-13, 0) {$D_1 :=$};
		\node [style=none] (27) at (-5.5, -4.75) {$y_m$};
		\node [style=none] (35) at (-7.5, 3) {$\vdots$};
		\node [style=white box, minimum width=2cm, minimum height=1cm] (36) at (-9, -5.5) {};
		\node [style=none] (37) at (-4, -3.5) {};
		\node [style=multiline text] (38) at (-9, -5.5) {TRUE and\\FALSE};
		\node [style=none] (39) at (-5.5, -0.5) {$y_1$};
		\node [style=white box, minimum width=2cm, minimum height=1cm] (40) at (-9, -1) {};
		\node [style=none] (41) at (-4, -1) {};
		\node [style=multiline text] (42) at (-9, -1) {TRUE and\\FALSE};
		\node [style=white box, minimum width=2cm, minimum height=1cm] (43) at (9, 0) {};
		\node [style=none] (44) at (9, 0) {Is TRUE?};
	\end{pgfonlayer}
	\begin{pgfonlayer}{edgelayer}
		\draw (4.center) to (7.center);
		\draw (5.center) to (2.center);
		\draw (37.center) to (36);
		\draw [in=360, out=180] (41.center) to (40);
		\draw (43) to (0);
	\end{pgfonlayer}
\end{tikzpicture}\\[8ex]
        \begin{tikzpicture}
	\begin{pgfonlayer}{nodelayer}
		\node [style=white box, minimum width=4cm, minimum height=4cm] (0) at (0, 0) {};
		\node [style=none] (1) at (-11, 5.5) {};
		\node [style=none] (2) at (-11, 1) {};
		\node [style=none] (3) at (-4, 3.5) {};
		\node [style=none] (4) at (-4, 1) {};
		\node [style=none] (5) at (-5.5, -2.25) {$\vdots$};
		\node [style=none] (6) at (-7.5, 5) {$x_1$};
		\node [style=none] (7) at (-7.5, 1.25) {$x_n$};
		\node [style=multiline text] (8) at (0, 0) {\Large$D_{\rho}$};
		\node [style=none] (9) at (-13, 0) {$D_2 :=$};
		\node [style=none] (10) at (-5.5, -4.75) {$z_m$};
		\node [style=none] (11) at (-7.5, 3) {$\vdots$};
		\node [style=white box, minimum width=2cm, minimum height=1cm] (12) at (-9, -5.5) {};
		\node [style=none] (13) at (-4, -3.5) {};
		\node [style=multiline text] (14) at (-9, -5.5) {TRUE and\\FALSE};
		\node [style=none] (15) at (-5.5, -0.5) {$z_1$};
		\node [style=white box, minimum width=2cm, minimum height=1cm] (16) at (-9, -1) {};
		\node [style=none] (17) at (-4, -1) {};
		\node [style=multiline text] (18) at (-9, -1) {TRUE and\\FALSE};
		\node [style=white box, minimum width=2cm, minimum height=1cm] (19) at (9, 0) {};
		\node [style=none] (20) at (9, 0) {Is TRUE?};
	\end{pgfonlayer}
	\begin{pgfonlayer}{edgelayer}
		\draw (2.center) to (4.center);
		\draw (3.center) to (1.center);
		\draw (13.center) to (12);
		\draw [in=360, out=180] (17.center) to (16);
		\draw [in=360, out=180] (19) to (0);
	\end{pgfonlayer}
\end{tikzpicture}
    \end{center}

    The variables are not part of the diagram and are written only to help indicate the logical expressions corresponding to the wires.

    Let $\ket{v} = \ket{v_1 \dots v_n}$ be any state from the $n$ qubits computational basis. Then, $v$ gives a valuation of $x_1, \dots, x_n$, by taking $v(x_i) = False$ when $v_i = 0$ and $v(x_i) = True$ when $v_i = 1$. Let $\mathcal{Y} = y_1, \dots, y_m$. We get:
    \begin{align*}
        \diageval[D_1]\ket{v} &=
        \bra{1} \diageval[D_{\psi}] (\ket{v} \otimes (\sqrt{2}^m \ket{+}^{\otimes m})) \\&=
        \sum_{v' \colon \mathcal{Y} \to \{ 0,1\}} \bra{1} \diageval[D_{\psi}] (\ket{v} \otimes (\ket{v'}) \\&=
        \sum_{v' \colon \mathcal{Y} \to \{ 0,1\}} \bra{1} \ket{(v'(v(\psi))} \\&=
        \bra{1} (\csat(v(\psi), \mathcal{Y}) \ket{1} + (2^{m} - \csat(v(\psi), \mathcal{Y})) \ket{0}) \\&=
        \csat(v(\psi), \mathcal{Y})
    \end{align*}
    where $\ket{v'} = \ket{v'(y_1) \dots v'(y_m)}$. The first equality comes from the definition of `Is TRUE' block as $\bra{1}$, the second equality from expanding $\ket{+\dots+}$ as the sum of computational basis states, the third equality from combining valuations $v$ of $x_1,\dots,x_n$ and $v'$ of $y_1,\dots,y_m$ to a single valuation of $x_1,\dots,x_n,y_1,\dots,y_m$, and the fourth equality follows from lemma \ref{formulaattached}.

    Similarly, $\diageval[D_2]\ket{v} = \csat(v(\rho), z_1,\dots,z_m)$. Thus, the diagrams equalize on state $\ket{v}$ precisely when $\csat(v(\psi), y_1,\dots,y_m) = \csat(v(\rho), z_1,\dots,z_m)$, which ends the reduction and the entire proof.
\end{proof}

\begin{example}
    Consider $\SATandComparecSAT$ instance from example \ref{satccsatexample}. The two diagrams $D_1$ and $D_2$ constructed in the proof of \ref{state equality} are presented below. All blocks are translated to the ZH generators. The dashed purple lines are not part of the diagram and are used only to the indicate logical expressions corresponding to the wires.
    
    
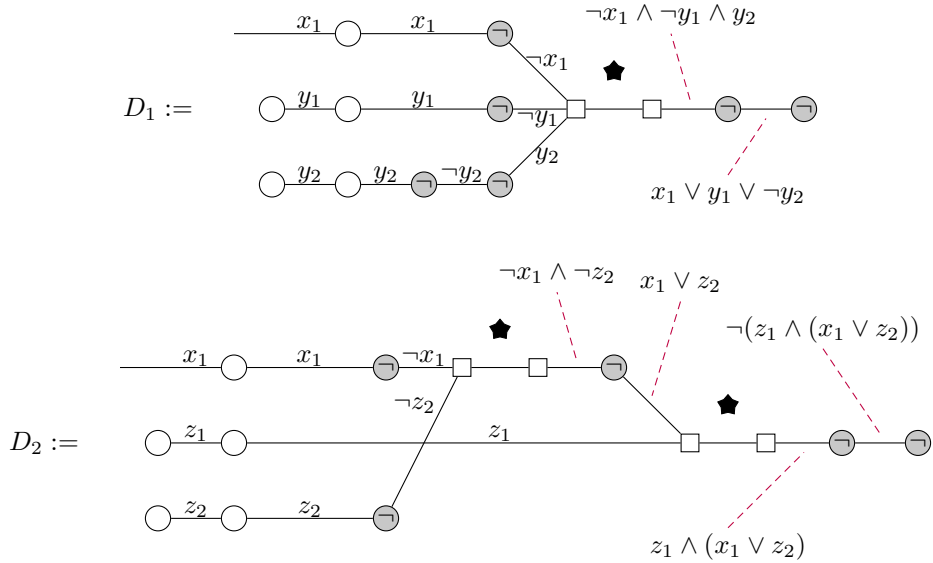
\begin{center}
        \begin{tikzpicture}
	\begin{pgfonlayer}{nodelayer}
		\node [style=none] (0) at (-7, 2) {};
		\node [style=white dot] (3) at (-4, 2) {};
		\node [style=white dot] (4) at (-4, 0) {};
		\node [style=white dot] (5) at (-4, -2) {};
		\node [style=white dot] (6) at (-6, 0) {};
		\node [style=white dot] (7) at (-6, -2) {};
		\node [style=dark] (8) at (0, 2) {};
		\node [style=dark] (9) at (0, 0) {};
		\node [style=dark] (10) at (0, -2) {};
		\node [style=white box] (11) at (2, 0) {};
		\node [style=white box] (12) at (4, 0) {};
		\node [style=dark] (13) at (6, 0) {};
		\node [style=dark] (14) at (8, 0) {};
		\node [style=halfscalar] (15) at (3, 1) {};
		\node [style=none] (16) at (0, 2) {$\neg$};
		\node [style=none] (17) at (0, 0) {$\neg$};
		\node [style=none] (18) at (0, -2) {$\neg$};
		\node [style=none] (19) at (6, 0) {$\neg$};
		\node [style=none] (20) at (8, 0) {$\neg$};
		\node [style=none] (21) at (-5, 2.25) {$x_1$};
		\node [style=none] (22) at (-5, 0.25) {$y_1$};
		\node [style=none] (23) at (-5, -1.75) {$y_2$};
		\node [style=none] (24) at (-2, 2.25) {$x_1$};
		\node [style=none] (25) at (-2, 0.25) {$y_1$};
		\node [style=none] (26) at (-3, -1.75) {$y_2$};
		\node [style=none] (27) at (1.25, -1.25) {$y_2$};
		\node [style=none] (28) at (1, -0.25) {$\neg y_1$};
		\node [style=none] (29) at (1.25, 1.25) {$\neg x_1$};
		\node [style=none] (30) at (4.5, 2.5) {$\neg x_1 \wedge \neg y_1 \wedge y_2$};
		\node [style=none] (31) at (6, -2.25) {$x_1 \vee y_1 \vee \neg y_2$};
		\node [style=dark] (32) at (-2, -2) {};
		\node [style=none] (33) at (-2, -2) {$\neg$};
		\node [style=none] (34) at (-1, -1.75) {$\neg y_2$};
		\node [style=none] (35) at (4.5, 2) {};
		\node [style=none] (36) at (5, 0.25) {};
		\node [style=none] (37) at (7, -0.25) {};
		\node [style=none] (38) at (6, -1.75) {};
		\node [style=none] (39) at (-2, 4) {};
		\node [style=none] (40) at (-2, -4) {};
		\node [style=none] (41) at (-9, 0) {$D_1 :=$};
		\node [style=none] (42) at (9, 0) {};
		\node [style=none] (43) at (-10, 0) {};
	\end{pgfonlayer}
	\begin{pgfonlayer}{edgelayer}
		\draw (0.center) to (3);
		\draw (3) to (8);
		\draw (8) to (11);
		\draw (11) to (9);
		\draw (10) to (11);
		\draw (11) to (12);
		\draw (12) to (13);
		\draw (13) to (14);
		\draw (6) to (4);
		\draw (4) to (9);
		\draw (5) to (7);
		\draw (5) to (32);
		\draw (32) to (10);
		\draw [style=purple dashed] (36.center) to (35.center);
		\draw [style=purple dashed] (38.center) to (37.center);
	\end{pgfonlayer}
\end{tikzpicture}\\
        \begin{tikzpicture}
	\begin{pgfonlayer}{nodelayer}
		\node [style=none] (0) at (-7, 2) {};
		\node [style=white dot] (1) at (-4, 2) {};
		\node [style=white dot] (2) at (-4, 0) {};
		\node [style=white dot] (3) at (-4, -2) {};
		\node [style=white dot] (4) at (-6, 0) {};
		\node [style=white dot] (5) at (-6, -2) {};
		\node [style=dark] (6) at (0, 2) {};
		\node [style=white box] (9) at (2, 2) {};
		\node [style=white box] (10) at (4, 2) {};
		\node [style=dark] (11) at (6, 2) {};
		\node [style=halfscalar] (13) at (3, 3) {};
		\node [style=none] (14) at (0, 2) {$\neg$};
		\node [style=none] (17) at (6, 2) {$\neg$};
		\node [style=none] (19) at (-5, 2.25) {$x_1$};
		\node [style=none] (20) at (-5, 0.25) {$z_1$};
		\node [style=none] (21) at (-5, -1.75) {$z_2$};
		\node [style=none] (22) at (-2, 2.25) {$x_1$};
		\node [style=none] (24) at (-2, -1.75) {$z_2$};
		\node [style=none] (25) at (0.75, 1) {$\neg z_2$};
		\node [style=none] (27) at (1, 2.25) {$\neg x_1$};
		\node [style=none] (28) at (4.5, 4.5) {$\neg x_1 \wedge \neg z_2$};
		\node [style=none] (29) at (7.75, 4.25) {$x_1 \vee z_2$};
		\node [style=dark] (30) at (0, -2) {};
		\node [style=none] (31) at (0, -2) {$\neg$};
		\node [style=none] (33) at (4.5, 4) {};
		\node [style=none] (34) at (5, 2.25) {};
		\node [style=none] (35) at (7, 1.25) {};
		\node [style=none] (36) at (7.75, 3.75) {};
		\node [style=none] (37) at (-2, 4) {};
		\node [style=none] (38) at (-2, -4) {};
		\node [style=none] (39) at (-9, 0) {$D_2 :=$};
		\node [style=none] (40) at (15, 0) {};
		\node [style=dark] (41) at (12, 0) {};
		\node [style=none] (42) at (12, 0) {$\neg$};
		\node [style=white box] (44) at (8, 0) {};
		\node [style=white box] (45) at (10, 0) {};
		\node [style=halfscalar] (46) at (9, 1) {};
		\node [style=none] (47) at (3, 0.25) {$z_1$};
		\node [style=none] (48) at (9, -2.75) {$z_1\wedge(x_1 \vee z_2)$};
		\node [style=none] (49) at (11, -0.25) {};
		\node [style=none] (50) at (9, -2.25) {};
		\node [style=dark] (51) at (14, 0) {};
		\node [style=none] (52) at (14, 0) {$\neg$};
		\node [style=none] (53) at (11.5, 3) {$\neg (z_1\wedge(x_1 \vee z_2))$};
		\node [style=none] (54) at (13, 0.25) {};
		\node [style=none] (55) at (11.5, 2.5) {};
		\node [style=none] (56) at (-10, 0) {};
	\end{pgfonlayer}
	\begin{pgfonlayer}{edgelayer}
		\draw (0.center) to (1);
		\draw (1) to (6);
		\draw (6) to (9);
		\draw (9) to (10);
		\draw (10) to (11);
		\draw (4) to (2);
		\draw (3) to (5);
		\draw (3) to (30);
		\draw [style=purple dashed] (34.center) to (33.center);
		\draw [style=purple dashed] (36.center) to (35.center);
		\draw (30) to (9);
		\draw (44) to (45);
		\draw (2) to (44);
		\draw (11) to (44);
		\draw (45) to (41);
		\draw [style=purple dashed] (50.center) to (49.center);
		\draw (41) to (51);
		\draw [style=purple dashed] (55.center) to (54.center);
	\end{pgfonlayer}
\end{tikzpicture}
    \end{center}

    The matrix representations of the diagrams correspond to the values found for different assignments of $x_1$ in example \ref{satccsatexample}:
    \begin{gather*}
        \diageval[D_1] = \begin{pmatrix}
            3 & 4
        \end{pmatrix},\qquad
        \diageval[D_2] = \begin{pmatrix}
            3 & 2
        \end{pmatrix}
    \end{gather*}
    and they indeed agree on the first entry.
\end{example}

Note, that the above constructions require only the blocks representing TRUE and FALSE simultaneously, COPY, NOT, AND gates, and Is TRUE check. It means, that $\StateEq$ problem is $\npsharpp$-hard not only over phase-free ZH but also over any language in which these blocks can be represented like ZX, ZW, etc. In particular, the version of $\StateEq$ problem where the input consists of two tensors constructed from the mentioned blocks is also $\npsharpp$-complete.

\begin{corollary}
    The problem $\ContainsEntry$ is $\npsharpp$-hard for all generator sets in which the following tensors can be achieved:
    \begin{gather*}
    \begin{pmatrix}
        1\\
        1
    \end{pmatrix},
    \begin{pmatrix}
        1 & 0\\
        0 & 0\\
        0 & 0\\
        0 & 1
    \end{pmatrix},
    \begin{pmatrix}
        0 & 1\\
        1 & 0
    \end{pmatrix},
    \begin{pmatrix}
        1 & 1 & 1 & 0\\
        0 & 0 & 0 & 1
    \end{pmatrix},\ \text{and}\ 
    \begin{pmatrix}
        0 & 1
    \end{pmatrix}.
\end{gather*}
\end{corollary}

All entries in these matrices are $0$ or $1$, so any tensor constructed from these has matrix representation with all entries being natural numbers. This is even stricter than phase-free ZH diagrams, where entries can be dyadic rationals due to the star generator.

\begin{theorem}\label{contains entry}
    $\ContainsEntryk$ is $\npcep$-hard.
\end{theorem}
\begin{proof}
    We reduce from $\SATandComparecSAT$. First, we focus on the cases $k\in \{0,1\}$. Let $(m,n,\psi,\rho)$ be any instance. The goal is to construct a diagram that when evaluated on a state corresponding to valuation $v$ of $x_1, \dots, x_n$ results in the number $\csat(v(\rho), y_1,\dots,y_m) - \csat(v(\psi), z_1,\dots,z_m) + k$. I.e. the resulting number equals $k$ precisely when $v$ is a witness for the initial $\SATandComparecSAT$ instance.
    
    In the previous proof, we showed how to get the numbers $\csat(v(\rho))$ and $\csat(v(\psi))$ in phase-free ZH. The difference can be achieved by utilizing white and dark NOTs. The addition of $k$ can be achieved by modifying the pair of initial boolean formulae. Consider the following diagram:

    
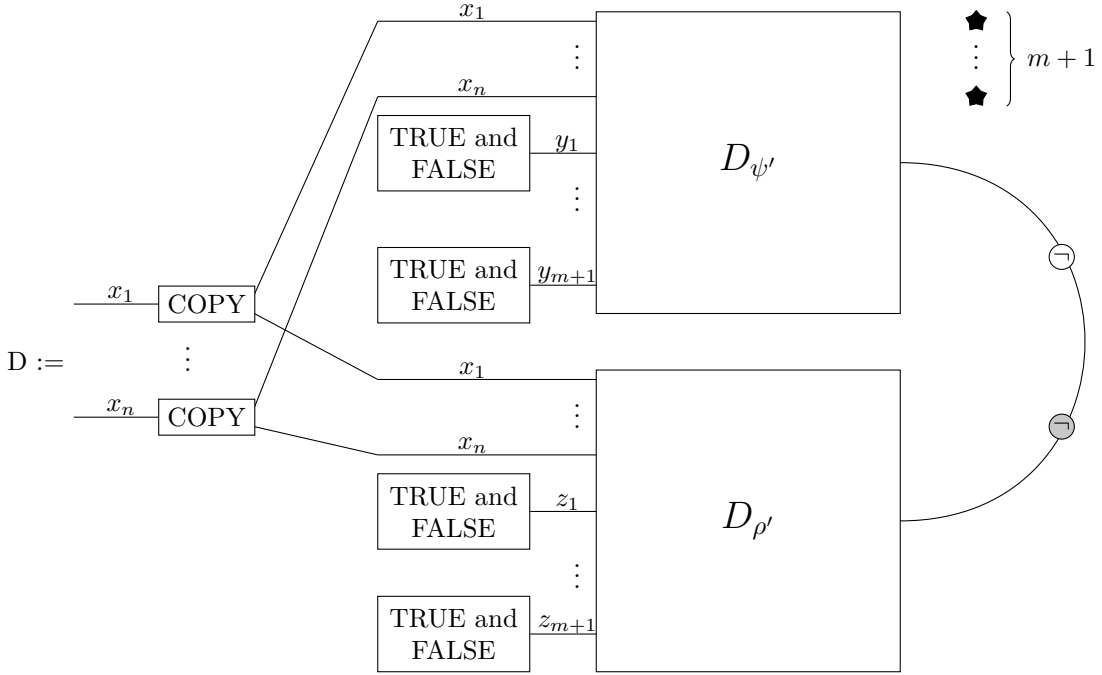
\begin{center}
        \begin{tikzpicture}
	\begin{pgfonlayer}{nodelayer}
		\node [style=white box, minimum width=4cm, minimum height=4cm] (0) at (-2, 5.25) {};
		\node [style=none] (1) at (-11.75, 9) {};
		\node [style=none] (2) at (-11.75, 7) {};
		\node [style=none] (3) at (-6, 9) {};
		\node [style=none] (4) at (-6, 7) {};
		\node [style=none] (5) at (-6.5, 4.5) {$\vdots$};
		\node [style=none] (6) at (-9.25, 9.25) {$x_1$};
		\node [style=none] (7) at (-9.25, 7.25) {$x_n$};
		\node [style=multiline text] (8) at (-2, 5.25) {\Large$D_{\psi'}$};
		\node [style=none] (10) at (-6.75, 2.25) {$y_{m+1}$};
		\node [style=none] (11) at (-6.5, 8.25) {$\vdots$};
		\node [style=white box, minimum width=2cm, minimum height=1cm] (12) at (-9.75, 2) {};
		\node [style=none] (13) at (-6, 2) {};
		\node [style=multiline text] (14) at (-9.75, 2) {TRUE and\\FALSE};
		\node [style=none] (15) at (-6.75, 5.75) {$y_1$};
		\node [style=white box, minimum width=2cm, minimum height=1cm] (16) at (-9.75, 5.5) {};
		\node [style=none] (17) at (-6, 5.5) {};
		\node [style=multiline text] (18) at (-9.75, 5.5) {TRUE and\\FALSE};
		\node [style=none, minimum width=2cm, minimum height=1cm] (19) at (2, 5.25) {};
		\node [style=white box, minimum width=4cm, minimum height=4cm] (21) at (-2, -4.25) {};
		\node [style=none] (22) at (-11.75, -0.5) {};
		\node [style=none] (23) at (-11.75, -2.5) {};
		\node [style=none] (24) at (-6, -0.5) {};
		\node [style=none] (25) at (-6, -2.5) {};
		\node [style=none] (26) at (-6.5, -5.5) {$\vdots$};
		\node [style=none] (27) at (-9.25, -0.25) {$x_1$};
		\node [style=none] (28) at (-9.25, -2.25) {$x_n$};
		\node [style=multiline text] (29) at (-2, -4.25) {\Large$D_{\rho'}$};
		\node [style=none] (31) at (-6.75, -7) {$z_{m+1}$};
		\node [style=none] (32) at (-6.5, -1.25) {$\vdots$};
		\node [style=white box, minimum width=2cm, minimum height=1cm] (33) at (-9.75, -7.25) {};
		\node [style=none] (34) at (-6, -7.25) {};
		\node [style=multiline text] (35) at (-9.75, -7.25) {TRUE and\\FALSE};
		\node [style=none] (36) at (-6.75, -3.75) {$z_1$};
		\node [style=white box, minimum width=2cm, minimum height=1cm] (37) at (-9.75, -4) {};
		\node [style=none] (38) at (-6, -4) {};
		\node [style=multiline text] (39) at (-9.75, -4) {TRUE and\\FALSE};
		\node [style=none, minimum width=2cm, minimum height=1cm] (40) at (2, -4.25) {};
		\node [style=none] (42) at (-19.75, 1.5) {};
		\node [style=none] (43) at (-19.75, -1.5) {};
		\node [style=white box] (44) at (-16.25, 1.5) {COPY};
		\node [style=white box] (45) at (-16.25, -1.5) {COPY};
		\node [style=none] (46) at (-18.5, 1.75) {$x_1$};
		\node [style=none] (47) at (-18.5, -1.25) {$x_n$};
		\node [style=none] (48) at (-16.75, 0.25) {$\vdots$};
		\node [style=none] (49) at (-15, 1.75) {};
		\node [style=none] (50) at (-15, -1.25) {};
		\node [style=none] (51) at (-15, 1.25) {};
		\node [style=none] (52) at (-15, -1.75) {};
		\node [style=white dot] (53) at (6.25, 2.75) {};
		\node [style=dark] (54) at (6.25, -1.75) {};
		\node [style=none] (55) at (6.25, 2.75) {$\neg$};
		\node [style=none] (56) at (6.25, -1.75) {$\neg$};
		\node [style=none] (57) at (-20.75, 0) {D :=};
		\node [style=halfscalar] (58) at (4, 9) {};
		\node [style=halfscalar] (59) at (4, 7) {};
		\node [style=none] (60) at (4, 8.25) {$\vdots$};
		\node [style=none] (61) at (5, 6.75) {};
		\node [style=none] (62) at (5, 9.25) {};
		\node [style=none] (63) at (6.25, 8) {$m+1$};
	\end{pgfonlayer}
	\begin{pgfonlayer}{edgelayer}
		\draw (2.center) to (4.center);
		\draw (3.center) to (1.center);
		\draw [in=360, out=180] (13.center) to (12);
		\draw [in=360, out=180] (17.center) to (16);
		\draw (19.center) to (0);
		\draw (23.center) to (25.center);
		\draw (24.center) to (22.center);
		\draw [in=360, out=180] (34.center) to (33);
		\draw [in=360, out=180] (38.center) to (37);
		\draw (40.center) to (21);
		\draw (43.center) to (45);
		\draw (44) to (42.center);
		\draw (1.center) to (49.center);
		\draw (2.center) to (50.center);
		\draw (22.center) to (51.center);
		\draw (52.center) to (23.center);
		\draw [bend left=90, looseness=1.75] (19.center) to (40.center);
		\draw [style=brace edge] (62.center) to (61.center);
	\end{pgfonlayer}
\end{tikzpicture}
    \end{center}

    Where $\psi'$ and $\rho'$ are defined as follows:
    \begin{align*}
        \psi' &= (\psi \wedge \neg y_{m+1})\\
        \rho' &= \begin{cases}
            (\rho \wedge \neg z_{m+1}),\quad &\text{when $k=0$}\\
            (\rho \wedge \neg z_{m+1}) \vee (z_1 \wedge \dots \wedge z_m \wedge z_{m+1}),\quad &\text{when $k=1$}
        \end{cases}
    \end{align*}

    Again, let $\ket{v} = \ket{v_1 \dots v_n}$ be any state from the $n$ qubits computational basis. Then, $v$ gives a valuation of $x_1, \dots, x_n$, by taking $v(x_i) = False$ when $v_i = 0$ and $v(x_i) = True$ when $v_i = 1$. Let $\mathcal{Y}' = y_1 \dots y_m y_{m+1}$ and $\mathcal{Z}' = z_1 \dots z_m z_{m+1}$.

    By considering requirements on variables $y_{m+1}$ and $z_{m+1}$, we obtain:
    \begin{align*}
        \csat(v(\psi'), y_1,\dots,y_{m+1}) &= \csat(v(\psi), y_1,\dots,y_m)\\
        \csat(v(\rho'), z_1,\dots,z_{m+1}) &= \left. \begin{cases}
            \csat(v(\rho), z_1,\dots,z_m),\quad &\text{when $k = 0$}\\
            \csat(v(\rho), z_1,\dots,z_m) + 1,\quad &\text{when $k = 1$}
        \end{cases} \right\}\\ &= \csat(v(\rho), z_1,\dots,z_m) + k
    \end{align*}
    
    Hence:
    \begin{multline*}
        \csat(v(\rho'), z_1,\dots,z_{m+1}) - \csat(v(\psi'), y_1,\dots,y_{m+1}) \\=
        \csat(v(\rho), z_1,\dots,z_m) - \csat(v(\psi), y_1,\dots,y_m) + k
    \end{multline*}

    The rightmost part of the diagram $D$ (the two nots on a cap) has matrix representation $\ket{01} - \ket{10}$. It is the only part of the proof where white not is used.
    
    Similarly to the previous proof, we get:
    \begin{align*}
        \diageval[D]\ket{v} &=
        \frac{1}{2^{m+1}} (\ket{01} - \ket{10})(\diageval[D_{\psi'}] \otimes \diageval[D_{\rho'}]) (\ket{v} \otimes \sqrt{2}^{m+1} \ket{+}^{\otimes (m+1)})^{\otimes 2} \\&=
        \frac{1}{2^{m+1}} (\ket{01} - \ket{10}) (a_1 \ket{1} + a_0 \ket{0}) \otimes (b_1 \ket{1} + b_0 \ket{0})\\
        \noalign{\centering $\text{where } \begin{cases}
            a_1 = (\csat(v(\phi'), \mathcal{Y}')\\
            a_0 = (2^{m+1} - \csat(v(\phi'), \mathcal{Y}'))\\
            b_1 = (\csat(v(\rho'), \mathcal{Z}')\\
            b_0 = (2^{m+1} - \csat(v(\rho'), \mathcal{Z}'))
        \end{cases}$}
    \end{align*}
    Therefore:
    \begin{align*}
        \diageval[D]\ket{v} &=
        \frac{1}{2^{m+1}} (a_0 b_1 - a_1 b_0) \\&=
        \frac{1}{2^{m+1}} ((2^{m+1} - a_1)b_1 - a_1(2^{m+1} - b_1)) \\&=
        \frac{1}{2^{m+1}} (2^{m+1} b_1 - a_1b_1 - 2^{m+1}a_1 + a_1b_1) \\&=
        b_1 - a_1 \\&=
        (\csat(v(\rho'), \mathcal{Z}') - (\csat(v(\phi'), \mathcal{Y}') \\&=
        \csat(v(\rho), z_1,\dots,z_m) - \csat(v(\psi), y_1,\dots,y_m) + k
    \end{align*}
    The last number equals $k$ precisely when $\csat(v(\psi), y_1,\dots,y_m) = \csat(v(\rho), z_1,\dots,z_m)$. Therefore instance of $\ContainsEntryk$ consisting of diagram $D$ is equivalent to the initial instance $(m,n,\psi,\rho)$ of $\SATandComparecSAT$.

    This ends the proof for $k \in \{ 0, 1\}$. For any other dyadic rational $k = \frac{c}{2^d}$, we can modify diagram $D$ for $k=1$ to also contain a scalar representing the new $k$ in phase-free ZH, effectively rescaling the diagram by a factor of $k$. This scalar can be obtained, for instance, by adding $d$ star generators and evaluating the formula with $c$ satisfying assignments on all of its assignments. Such a formula can be constructed by lemma \ref{number to formula}. Hence, the problem $\ContainsEntryk$ is $\npsharpp$-hard for all dyadic rationals $k$.
\end{proof}

Note, that when $k$ is not a dyadic rational, then $k$ cannot appear in the matrix representation of any phase-free ZH diagram. Then, $\ContainsEntryk$ problem is trivial -- any instance can be immediately answered with $False$.

In the above proof, we had to obtain a difference of two numbers in phase-free ZH and rescale such difference by a negative power of two. We used white not and star generators. It extends the necessary tensors over those for $\StateEq$ by tensors with matrix representation:
\begin{gather*}
    \begin{pmatrix}
        1 & 0\\
        0 & -1
    \end{pmatrix},\ \text{and }\frac{1}{2}.
\end{gather*}
The second of these is not necessary when $k = 0$. Thus, the problem $\ContainsEntry_0$ is $\npsharpp$-hard even for tensors with matrix representations containing only integers.

Combining everything from this section, we obtain the proof of the main theorem:

\begingroup\def\thetheorem{\ref{main result}}\begin{theorem}[Repeated]
    $\StateEq$ and $\ContainsEntryk$ are both $\npsharpp$-complete.
\end{theorem}\addtocounter{theorem}{-1}\endgroup

\begin{proof}
    \sloppy Completeness of $\StateEq$ follows from theorems \ref{cont stateeq} and \ref{state equality}, and completeness of $\ContainsEntryk$ follows from theorems \ref{cont containsentry} and \ref{contains entry}.
\end{proof}

\section{Circuit Extraction}\label{circuitextractionsection}
As an additional result, we present how to adapt proof form \cite{debeaudrapCircuitExtractionZXDiagrams2022c} for phase-free ZH, i.e. we show that circuit extraction remains $\counting{P}$-hard for phase-free ZH diagrams.
The original proof required a $\frac{\pi}{2}$ phase spider which cannot be constructed in phase-free ZH.

\begingroup\def\thetheorem{\ref{circuit extraction hardness}}\begin{theorem}[Repeated]
    $\CircuitExtraction$ is $\counting{P}$-hard.
\end{theorem}\addtocounter{theorem}{-1}\endgroup

\begin{proof}
    We reduce from $\csat$. Let $\phi$ be any boolean formula and $x_1, \dots, x_n$ be the variables on which it is considered. Consider the following diagram:
    
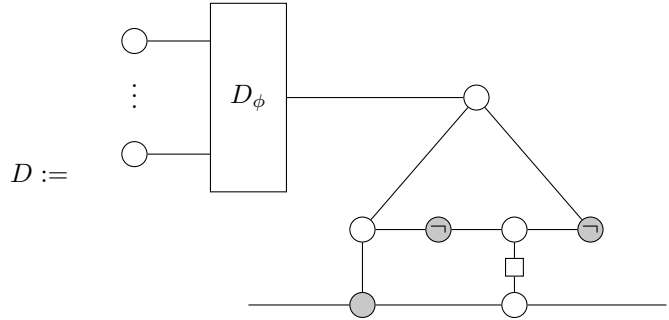
\begin{equation*}
        \begin{tikzpicture}
	\begin{pgfonlayer}{nodelayer}
		\node [style=white box, minimum width=1cm, minimum height=2.5cm] (0) at (0, -2.5) {$D_{\phi}$};
		\node [style=white dot] (1) at (-3, -1) {};
		\node [style=white dot] (3) at (-3, -4) {};
		\node [style=white dot] (4) at (6, -2.5) {};
		\node [style=none] (5) at (0, -1) {};
		\node [style=none] (7) at (0, -4) {};
		\node [style=none] (9) at (-3, -2.25) {$\vdots$};
		\node [style=white dot] (10) at (3, -6) {};
		\node [style=dark] (11) at (5, -6) {};
		\node [style=dark] (12) at (9, -6) {};
		\node [style=white dot] (13) at (7, -6) {};
		\node [style=dark] (14) at (3, -8) {};
		\node [style=white dot] (15) at (7, -8) {};
		\node [style=none] (16) at (0, -8) {};
		\node [style=none] (17) at (11, -8) {};
		\node [style=none] (18) at (5, -6) {$\neg$};
		\node [style=none] (19) at (9, -6) {$\neg$};
		\node [style=white box] (20) at (7, -7) {};
		\node [style=none] (21) at (-5.5, -4.5) {$D :=$};
	\end{pgfonlayer}
	\begin{pgfonlayer}{edgelayer}
		\draw (1) to (5.center);
		\draw (3) to (7.center);
		\draw (0) to (4);
		\draw (10) to (14);
		\draw (10) to (11);
		\draw (11) to (13);
		\draw (13) to (12);
		\draw (12) to (4);
		\draw (4) to (10);
		\draw (14) to (15);
		\draw (15) to (17.center);
		\draw (14) to (16.center);
		\draw (13) to (20);
		\draw (20) to (15);
	\end{pgfonlayer}
\end{tikzpicture}
    \end{equation*}

    By lemma \ref{formulaattached}, $\diageval[D_{\phi}] \sqrt{2}^n \ket{+}^{\otimes n} = a_0 \ket{0} + a_1 \ket{1}$, where:
    \begin{align*}
        a_1 &= \csat(\phi)\\
        a_0 &= 2^n - \csat(\phi)
    \end{align*}
    By considering the matrix representation of the remaining part of diagram $D$ we can show that:
    \begin{gather*}
        \diageval[D] = \begin{bmatrix}
            a_0 & a_1\\
            a_1 & -a_0
        \end{bmatrix} =: M
    \end{gather*}
    which is proportional to a unitary:
    \begin{gather*}
        MM^{\dagger} = \begin{bmatrix}
            a_0^2 + a_1^2 & 0\\
            0 & a_0^2 + a_1^2
        \end{bmatrix}
    \end{gather*}
    We show how to find $\diageval[D]$ formally in the appendix \ref{circuit extraction for phase free zh applied ket}.

    The finish of the proof is the same as in \cite{debeaudrapCircuitExtractionZXDiagrams2022c}. Given a polysize ancilla-free circuit equivalent to the presented diagram, we can compute the matrix representation of the circuit in polynomial time. The obtained matrix representation is proportional to matrix $M$ above. Hence, it is possible to deduce $\csat(\phi)$. Therefore $\csat \in {\RM{FP}}^{\CircuitExtraction}$, which means precisely that $\CircuitExtraction$ is $\counting{P}$-hard.
\end{proof}

Our construction is a bit more complex than the one from the proof for ZX calculus -- rather than applying the controlled $iX$ gate, we apply the variant of the controlled $iY$ gate that can be achieved in phase-free ZH. It has the advantage of working in every graphical calculus in which one can represent $\diageval[D_{\phi}]$, Hadamard gate, and equivalents of $\pi$ phase spiders from ZX. Similarly to \cite{debeaudrapCircuitExtractionZXDiagrams2022c}, the hardness of different variants of circuit extraction follows from the same construction as above.

\section{Conclusions and further work}\label{lastsection}
We gave two examples of $\npsharpp$-complete problems arising in phase-free ZH. The main challenge was crafting a suitable $\npsharpp$-complete problem that could be used for reduction. To define such a problem, we showed that $\npsharpp = \npcep$ and proceeded in a way similar to the proof of the Cook-Levin theorem.

\subsection{Comparing Diagrams}
Our results could suggest that comparing diagrams should be $\rm{coNP}^{\counting{P}}$-complete. At the same time, comparing quantum circuits is simpler, in the sense that upper bounds are tighter. Determining whether a quantum circuit is almost equivalent to identity is $\rm{QMA}$-complete \cite{janzingNonidentitycheckQmacomplete2005} (Quantum Merlin-Arthur \cite{watrousSuccinctQuantumProofs2000}), while an exact non-identity check is $\RM{NQP}$-complete \cite{tanakaEXACTNONIDENTITYCHECK2010} (Nondeterministic Quantum Polynomial-time \cite{adlemanQuantumComputability1997}). Both these classes are within $\RM{PP}$, and thus within $\RM{P}^{\counting{P}}$. This way, graphical calculi like phase-free ZH can be understood as more powerful than quantum circuits.

The problems we studied closely relate to problems $\CompareDiagrams$ and the following problem:

\begin{quote}
    $\IsZero$\\
    \textbf{Input: } a phase-free ZH diagram $D$.\\
    \textbf{Output:} $True$ if and only if $\diageval[D]$ is the zero matrix.
\end{quote}

The two problems defined above are essential. Not only the question \textit{``Are two diagrams equal''} is the most natural question one can ask about diagrams, but the problem $\CompareDiagrams$ was indirectly used to obtain the upper bound of $\CircuitExtraction$ in \cite{debeaudrapCircuitExtractionZXDiagrams2022c}, i.e. $\CircuitExtraction$ is in $\RM{FNP}^{\CompareDiagrams}$.

These problems are very similar to the two problems we have presented. $\StateEq$ asks whether the two matrix representations agree \textit{somewhere}, while $\CompareDiagrams$ asks whether the two matrices agree \textit{everywhere}. Similarly, $\ContainsEntry\mathbf{_0}$ asks whether the matrix agrees with zero matrix \textit{somewhere}, and $\IsZero$ asks for agreement \textit{everywhere}. The problems defined here are therefore asking whether a property is \textit{universal} rather than \textit{existential}. The upper bounds for these problems follow a similar idea as the upper bounds for $\StateEq$ and $\ContainsEntry$. A NDTM can pick an entry, and using the $\counting{P}$ oracle it accepts when the property holds. It leads to upper bounds $\RM{coNP}^{\counting{P}}$.

However, the same hardness proof does not work. The main trick in our hardness proof was that an NDTM with $\counting{P}$ oracle could instead guess answers to the $\cSAT$ instances and verify them at the end with a single $\cep$ oracle query. However, the transformation creates lots of new non-deterministic paths on which NDTM chooses at least one answer incorrectly. Informally, there is no method in ZH to \textit{use} answer to the question of whether boolean formulae have different numbers of satisfying assignments or to verify that two amplitudes are different. Similarly, it is unclear whether $\RM{coNP}^{\counting{P}}$ must equal to $\RM{coNP}^{\cep[1]}$ if for the second class we demand an oracle answer to be returned as the answer to the entire computation.

It would be interesting to research whether $\CompareDiagrams$ and $\IsZero$ are $\RM{coNP}^{\counting{P}}$-complete, or at least obtain hardness better than $\cep$. Such immediate hardness can be obtained by reducing $\phi, \psi$ instance of $\ComparecSAT$ to diagrams $D_{\phi}, D_{\psi}$ with white spiders attached to all dangling edges. One of the possible avenues is to pay more attention to known complexity results, such as $\RM{NQP} = \RM{coC_=P}$ \cite{fennerDeterminingAcceptancePossibility1999}, and use $\RM{NQP}$ oracle instead of $\RM{C_=P}$.

\subsection{Tensor networks and complexity classes}
When proving the hardness of $\StateEq$ and $\ContainsEntryk$, we touched on the equivalent problems defined for tensors instead of diagrams. It may be interesting to classify complexity bounds for tensors constructed with different generator sets. Tensors required for phase-free ZH are sufficient for hardness results. In particular, generators of phase-free ZH can all be achieved in other related graphical calculi, like ZX. The upper bound is a bit more tricky. For instance, whether the presented problems are $\npsharpp$-complete for ZX diagrams depends on the allowed phases and presentations of such phases in the input. The infinite number of generators can lead to $\ScalarDiagram$ not being in $\RM{FP}^{\counting{P}}$, and hence the $\npsharpp$ can also fail.

Besides tensor networks, there is an interesting connection to the class $\RM{PostBQP}$, i.e. quantum computation with post-selection of measurement outcomes. Graphical calculi can be viewed as circuits with post-selection of measurement outcomes, or as quantum circuits where arbitrary linear maps rather than unitaries can be used as gates. Thus, diagrams represent computation schemes for the class $\RM{PostBQP}$. Further, it is known that $\RM{PostBQP} = \RM{PP}$ and hence $\RM{NP}^{\RM{PostBQP}} = \RM{NP}^{\RM{PP}} = \npsharpp$, where the last equality follows from Toda's theorem \cite{todaPPHardPolynomialTime1991}. It would be interesting to see whether the approach of interpreting diagrams as $\RM{PostBQP}$ computations could lead to new complexity bounds for problems appearing when working with graphical calculi. Further, the study of the class $\RM{NP}^{\RM{PostBQP}}$ could lead to new observations about the equal $\npsharpp$. The problem $\SATandComparecSAT$ we crafted can be viewed as a new natural complete problem for the mostly unexplored class $\npsharpp$.

\subsection{Counting complexity}
In \cite{laakkonenGraphicalSATAlgorithm2022}, phase-free ZH was successfully used to develop a new algorithm for a variant of $\cSAT$ with small clause density, beating other existing algorithms. Another connection of phase-free ZH with counting complexity was demonstrated in \cite{laakkonenPicturingCountingReductions2023}, where ZH rewriting rules simplified proofs of known results from counting complexity. However, only scalar diagrams were used. It would be interesting to check whether the inclusion of dangling edges may be used for other counting complexity proofs, for instance, those involving $\npsharpp$.

\subsection{Acknowledgement}
I want to thank my supervisor Miriam Backens for helping me put this work together and sharing their essential knowledge in the field of counting problems. I also want to thank Niel de Beaudrap and John van de Wetering for useful discussions about their work regarding circuit extraction hardness.


\phantomsection

\addcontentsline{toc}{chapter}{Bibliography}

\begin{sloppypar}
\bibliographystyle{quantum}
\bibliography{myrefdbwin}
\end{sloppypar}


\appendix

\section{Complexity Classes}\label{cococlasses}
Here, we define all complexity classes relevant to the presented work.

\begin{definition}[Complexity Classes]$ $
    \begin{itemize}
        \item $\RM{P}$ -- the class of problems solvable by a polytime deterministic TM,
        \item $\RM{NP}$ -- the class of problems solvable by a polytime NDTM,
        \item $\RM{coNP}$ -- the class of complements $\bar{A}$ of problems $A \in \RM{NP}$,
        \item $\RM{FP},\ \RM{FNP}$ -- function extensions of $\RM{P}$ and $\RM{NP}$,
        \item $\counting{P}$ -- the class of problems of the form: given a polytime NDTM $M$ and an input $w$, compute $\apaths(M,w)$,
        \item $\cep$ -- the class of problems $A$ for which there exists a polytime NDTM $M$ with the property, that $w \in A$ if and only if $\apaths(M,w) = \rpaths(M,w)$ (equivalently $\apaths(M,w) = \frac{1}{2}\paths(M,w)$),
        \item $\npsharpp$ -- the class of problems solvable by a polytime NDTM with access to $\counting{P}$ oracle,
        \item $\npcep$ -- the class of problems solvable by a polytime NDTM with access to $\cep$ oracle that is used precisely once,
        \item $\RM{PP}$ -- the class of problems $A$ for which there exists a polytime NDTM $M$ with the property, that $w \in A$ if and only if $\apaths(M,w) \ge \frac{1}{2}\paths(M,w)$,
        \item $\RM{BQP}$ -- the class of problems solvable by a polytime quantum TM,
        \item $\RM{PostBQP}$ -- informally, the class of problems solvable by a polytime quantum TM with postselection of measurement outcomes,
        \item $\RM{NQP}$ -- the class of problems solvable by a non-deterministic quantum TM . It is equal to $\RM{co}\cep$.
    \end{itemize}
\end{definition}

\section{Evaluation of the diagram for circuit extraction hardness}\label{circuit extraction for phase free zh applied ket}

Graphical proof for evaluation of the diagram in the proof of circuit extraction hardness. Equalities are labelled by lemmas and axioms from \cite{backensCompletenessZHcalculus2023}. The final matrix follows from the definition of nots in phase-free ZH.

\begin{center}
\begin{adjustbox}{width=\columnwidth-1cm,center}
    \begin{tikzpicture}
	\begin{pgfonlayer}{nodelayer}
		\node [style=white dot] (0) at (-5, 2) {};
		\node [style=white dot] (1) at (-8, 0) {};
		\node [style=dark] (2) at (-6, 0) {};
		\node [style=dark] (3) at (-2, 0) {};
		\node [style=white dot] (4) at (-4, 0) {};
		\node [style=dark] (5) at (-8, -2) {};
		\node [style=white dot] (6) at (-4, -2) {};
		\node [style=none] (7) at (-10, -2) {};
		\node [style=none] (8) at (-1, -2) {};
		\node [style=none] (9) at (-6, 0) {$\neg$};
		\node [style=none] (10) at (-2, 0) {$\neg$};
		\node [style=dark] (11) at (-9, 2) {};
		\node [style=white box] (12) at (-4, -1) {};
		\node [style=none] (13) at (0.5, 0) {$=$};
		\node [style=white dot] (15) at (4, 0) {};
		\node [style=dark] (16) at (6, 0) {};
		\node [style=dark] (17) at (10, 0) {};
		\node [style=white dot] (18) at (8, 0) {};
		\node [style=dark] (19) at (4, -2) {};
		\node [style=white dot] (20) at (8, -2) {};
		\node [style=none] (21) at (2, -2) {};
		\node [style=none] (22) at (11, -2) {};
		\node [style=none] (23) at (6, 0) {$\neg$};
		\node [style=none] (24) at (10, 0) {$\neg$};
		\node [style=dark] (25) at (5.5, 1) {};
		\node [style=white box] (26) at (8, -1) {};
		\node [style=dark] (27) at (8.5, 1) {};
		\node [style=none] (28) at (12.5, 0) {$=$};
		\node [style=dark] (30) at (18, 0) {};
		\node [style=dark] (31) at (22, 0) {};
		\node [style=white dot] (32) at (20, 0) {};
		\node [style=dark] (33) at (16, -2) {};
		\node [style=white dot] (34) at (20, -2) {};
		\node [style=none] (35) at (14, -2) {};
		\node [style=none] (36) at (23, -2) {};
		\node [style=none] (37) at (18, 0) {$\neg$};
		\node [style=none] (38) at (22, 0) {$\neg$};
		\node [style=dark] (39) at (17, 0) {};
		\node [style=white box] (40) at (20, -1) {};
		\node [style=dark] (41) at (20.5, 1) {};
		\node [style=dark] (42) at (16, -1) {};
		\node [style=none] (43) at (-11.5, -6) {$=$};
		\node [style=dark] (44) at (-6, -5) {};
		\node [style=dark] (45) at (-2, -5) {};
		\node [style=white dot] (46) at (-4, -5) {};
		\node [style=dark] (47) at (-8, -7) {};
		\node [style=white dot] (48) at (-4, -7) {};
		\node [style=none] (49) at (-10, -7) {};
		\node [style=none] (50) at (-1, -7) {};
		\node [style=none] (51) at (-6, -5) {$\neg$};
		\node [style=none] (52) at (-2, -5) {$\neg$};
		\node [style=white box] (54) at (-4, -6) {};
		\node [style=none] (57) at (0.5, -6) {$=$};
		\node [style=dark] (58) at (8, -5) {};
		\node [style=dark] (59) at (10, -5) {};
		\node [style=dark] (61) at (4, -7) {};
		\node [style=white dot] (62) at (8, -7) {};
		\node [style=none] (63) at (2, -7) {};
		\node [style=none] (64) at (11, -7) {};
		\node [style=none] (65) at (8, -5) {$\neg$};
		\node [style=white box] (67) at (8, -6) {};
		\node [style=none] (68) at (12.5, -6) {$=$};
		\node [style=dark] (69) at (20, -5) {};
		\node [style=dark] (70) at (22, -5) {};
		\node [style=dark] (71) at (16, -7) {};
		\node [style=white dot] (72) at (20, -7) {};
		\node [style=none] (73) at (14, -7) {};
		\node [style=none] (74) at (23, -7) {};
		\node [style=white box] (76) at (20, -6) {};
		\node [style=white box] (77) at (19, -5) {};
		\node [style=white box] (78) at (18, -5) {};
		\node [style=halfscalar] (79) at (21, -6) {};
		\node [style=none] (80) at (-11.5, -11) {$=$};
		\node [style=dark] (82) at (-4, -10) {};
		\node [style=white dot] (84) at (-6, -12) {};
		\node [style=none] (85) at (-8, -12) {};
		\node [style=none] (86) at (-3, -12) {};
		\node [style=white box] (87) at (-6, -11) {};
		\node [style=white box] (88) at (-6, -10) {};
		\node [style=white box] (89) at (-7, -10) {};
		\node [style=halfscalar] (90) at (-5, -11) {};
		\node [style=none] (91) at (0.5, -11) {$=$};
		\node [style=white dot] (103) at (6, -12) {};
		\node [style=none] (104) at (4, -12) {};
		\node [style=none] (105) at (9, -12) {};
		\node [style=white box] (106) at (6, -11) {};
		\node [style=halfscalar] (107) at (7, -11) {};
		\node [style=white dot] (108) at (7, -10) {};
		\node [style=halfscalar] (109) at (8, -11) {};
		\node [style=white dot] (110) at (8, -10) {};
		\node [style=none] (111) at (12.5, -11) {$=$};
		\node [style=white dot] (112) at (18.5, -11) {};
		\node [style=none] (113) at (16.5, -11) {};
		\node [style=none] (114) at (20.5, -11) {};
		\node [style=none] (115) at (18.5, -11) {$\neg$};
		\node [style=none] (116) at (0.5, 1) {2.21};
		\node [style=none] (117) at (12.5, 1) {2.21};
		\node [style=multiline text] (118) at (-11.5, -5) {2.10\\2.13};
		\node [style=multiline text] (119) at (0.5, -5) {2.22\\2.12};
		\node [style=none] (120) at (-11.5, -10) {2.11};
		\node [style=none] (121) at (12.5, -5) {(NOT)};
		\node [style=multiline text] (122) at (12.5, -10) {(Z)\\2.3};
		\node [style=multiline text] (124) at (0.5, -10) {(hh)\\(X)};
	\end{pgfonlayer}
	\begin{pgfonlayer}{edgelayer}
		\draw (1) to (5);
		\draw (1) to (2);
		\draw (2) to (4);
		\draw (4) to (3);
		\draw (3) to (0);
		\draw (0) to (1);
		\draw (5) to (6);
		\draw (6) to (8.center);
		\draw (5) to (7.center);
		\draw (11) to (0);
		\draw (12) to (4);
		\draw (12) to (6);
		\draw (15) to (19);
		\draw (15) to (16);
		\draw (16) to (18);
		\draw (18) to (17);
		\draw (19) to (20);
		\draw (20) to (22.center);
		\draw (19) to (21.center);
		\draw (26) to (18);
		\draw (26) to (20);
		\draw (27) to (17);
		\draw (25) to (15);
		\draw (30) to (32);
		\draw (32) to (31);
		\draw (33) to (34);
		\draw (34) to (36.center);
		\draw (33) to (35.center);
		\draw (40) to (32);
		\draw (40) to (34);
		\draw (41) to (31);
		\draw (39) to (30);
		\draw (44) to (46);
		\draw (46) to (45);
		\draw (47) to (48);
		\draw (48) to (50.center);
		\draw (47) to (49.center);
		\draw (54) to (46);
		\draw (54) to (48);
		\draw (42) to (33);
		\draw (61) to (62);
		\draw (62) to (64.center);
		\draw (61) to (63.center);
		\draw (67) to (62);
		\draw (58) to (67);
		\draw (71) to (72);
		\draw (72) to (74.center);
		\draw (71) to (73.center);
		\draw (76) to (72);
		\draw (69) to (76);
		\draw (78) to (77);
		\draw (77) to (69);
		\draw (84) to (86.center);
		\draw (87) to (84);
		\draw (89) to (88);
		\draw (88) to (87);
		\draw (85.center) to (84);
		\draw (103) to (105.center);
		\draw (106) to (103);
		\draw (104.center) to (103);
		\draw (112) to (114.center);
		\draw (113.center) to (112);
	\end{pgfonlayer}
\end{tikzpicture}
\end{adjustbox}\vspace{2cm}
\begin{adjustbox}{width=\columnwidth-1cm,center}
    \begin{tikzpicture}
	\begin{pgfonlayer}{nodelayer}
		\node [style=white dot] (0) at (4, 3) {};
		\node [style=white dot] (1) at (1, 1) {};
		\node [style=dark] (2) at (3, 1) {};
		\node [style=dark] (3) at (7, 1) {};
		\node [style=white dot] (4) at (5, 1) {};
		\node [style=dark] (5) at (1, -1) {};
		\node [style=white dot] (6) at (5, -1) {};
		\node [style=none] (7) at (-1, -1) {};
		\node [style=none] (8) at (8, -1) {};
		\node [style=none] (9) at (3, 1) {$\neg$};
		\node [style=none] (10) at (7, 1) {$\neg$};
		\node [style=dark] (11) at (0, 3) {};
		\node [style=white box] (12) at (5, 0) {};
		\node [style=none] (13) at (9.5, 1) {$=$};
		\node [style=white dot] (14) at (13, 1) {};
		\node [style=dark] (15) at (15, 1) {};
		\node [style=dark] (16) at (19, 1) {};
		\node [style=white dot] (17) at (17, 1) {};
		\node [style=dark] (18) at (13, -1) {};
		\node [style=white dot] (19) at (17, -1) {};
		\node [style=none] (20) at (11, -1) {};
		\node [style=none] (21) at (20, -1) {};
		\node [style=none] (22) at (15, 1) {$\neg$};
		\node [style=none] (23) at (19, 1) {$\neg$};
		\node [style=dark] (24) at (14.5, 2) {};
		\node [style=white box] (25) at (17, 0) {};
		\node [style=dark] (26) at (17.5, 2) {};
		\node [style=none] (27) at (21.5, 1) {$=$};
		\node [style=dark] (28) at (27, 1) {};
		\node [style=dark] (29) at (31, 1) {};
		\node [style=white dot] (30) at (29, 1) {};
		\node [style=dark] (31) at (25, -1) {};
		\node [style=white dot] (32) at (29, -1) {};
		\node [style=none] (33) at (23, -1) {};
		\node [style=none] (34) at (32, -1) {};
		\node [style=none] (35) at (27, 1) {$\neg$};
		\node [style=none] (36) at (31, 1) {$\neg$};
		\node [style=dark] (37) at (26, 1) {};
		\node [style=white box] (38) at (29, 0) {};
		\node [style=dark] (39) at (29.5, 2) {};
		\node [style=dark] (40) at (25, 0) {};
		\node [style=none] (41) at (-2.5, -5) {$=$};
		\node [style=dark] (42) at (3, -4) {};
		\node [style=dark] (43) at (7, -4) {};
		\node [style=white dot] (44) at (5, -4) {};
		\node [style=dark] (45) at (1, -6) {};
		\node [style=white dot] (46) at (5, -6) {};
		\node [style=none] (47) at (-1, -6) {};
		\node [style=none] (48) at (8, -6) {};
		\node [style=white box] (51) at (5, -5) {};
		\node [style=none] (52) at (9.5, -5) {$=$};
		\node [style=dark] (53) at (17, -4) {};
		\node [style=dark] (54) at (19, -4) {};
		\node [style=dark] (55) at (13, -6) {};
		\node [style=white dot] (56) at (17, -6) {};
		\node [style=none] (57) at (11, -6) {};
		\node [style=none] (58) at (20, -6) {};
		\node [style=white box] (60) at (17, -5) {};
		\node [style=none] (61) at (21.5, -5) {$=$};
		\node [style=white box] (62) at (29, -4) {};
		\node [style=dark] (64) at (25, -6) {};
		\node [style=white dot] (65) at (29, -6) {};
		\node [style=none] (66) at (23, -6) {};
		\node [style=none] (67) at (32, -6) {};
		\node [style=white box] (68) at (29, -5) {};
		\node [style=white dot] (69) at (28, -4) {};
		\node [style=halfscalar] (71) at (30, -5) {};
		\node [style=none] (72) at (-2.5, -10) {$=$};
		\node [style=none] (95) at (9.5, 2) {2.22};
		\node [style=none] (96) at (21.5, 2) {2.22};
		\node [style=multiline text] (97) at (-2.5, -4) {2.12\\2.13};
		\node [style=multiline text] (98) at (9.5, -4) {2.21\\2.10};
		\node [style=none] (99) at (-2.5, -9) {(hh)};
		\node [style=none] (100) at (21.5, -4) {(X)};
		\node [style=none] (103) at (0, 3) {$\neg$};
		\node [style=none] (104) at (14.5, 2) {$\neg$};
		\node [style=none] (105) at (17.5, 2) {$\neg$};
		\node [style=none] (106) at (29.5, 2) {$\neg$};
		\node [style=none] (107) at (26, 1) {$\neg$};
		\node [style=none] (108) at (25, 0) {$\neg$};
		\node [style=none] (109) at (1, -6) {$\neg$};
		\node [style=none] (110) at (13, -6) {$\neg$};
		\node [style=none] (111) at (25, -6) {$\neg$};
		\node [style=dark] (114) at (1, -11) {};
		\node [style=white dot] (115) at (5, -11) {};
		\node [style=none] (116) at (-1, -11) {};
		\node [style=none] (117) at (8, -11) {};
		\node [style=white dot] (119) at (5, -10) {};
		\node [style=halfscalar] (120) at (6, -10) {};
		\node [style=none] (121) at (1, -11) {$\neg$};
		\node [style=none] (122) at (9.5, -10) {$=$};
		\node [style=multiline text] (123) at (9.5, -9) {(zs)\\2.3};
		\node [style=white dot] (124) at (6, -9) {};
		\node [style=white dot] (125) at (7, -9) {};
		\node [style=halfscalar] (126) at (7, -10) {};
		\node [style=dark] (127) at (14.5, -10) {};
		\node [style=none] (129) at (12.5, -10) {};
		\node [style=none] (130) at (18.5, -10) {};
		\node [style=none] (133) at (14.5, -10) {$\neg$};
		\node [style=white dot] (134) at (31, -4) {};
		\node [style=halfscalar] (135) at (31, -5) {};
		\node [style=white dot] (136) at (16.5, -10) {};
		\node [style=none] (137) at (21.5, -10) {$=$};
		\node [style=multiline text] (138) at (21.5, -9) {(id)};
		\node [style=dark] (139) at (27.5, -10) {};
		\node [style=none] (140) at (25.5, -10) {};
		\node [style=none] (141) at (29.5, -10) {};
		\node [style=none] (142) at (27.5, -10) {$\neg$};
	\end{pgfonlayer}
	\begin{pgfonlayer}{edgelayer}
		\draw (1) to (5);
		\draw (1) to (2);
		\draw (2) to (4);
		\draw (4) to (3);
		\draw (3) to (0);
		\draw (0) to (1);
		\draw (5) to (6);
		\draw (6) to (8.center);
		\draw (5) to (7.center);
		\draw (11) to (0);
		\draw (12) to (4);
		\draw (12) to (6);
		\draw (14) to (18);
		\draw (14) to (15);
		\draw (15) to (17);
		\draw (17) to (16);
		\draw (18) to (19);
		\draw (19) to (21.center);
		\draw (18) to (20.center);
		\draw (25) to (17);
		\draw (25) to (19);
		\draw (26) to (16);
		\draw (24) to (14);
		\draw (28) to (30);
		\draw (30) to (29);
		\draw (31) to (32);
		\draw (32) to (34.center);
		\draw (31) to (33.center);
		\draw (38) to (30);
		\draw (38) to (32);
		\draw (39) to (29);
		\draw (37) to (28);
		\draw (42) to (44);
		\draw (44) to (43);
		\draw (45) to (46);
		\draw (46) to (48.center);
		\draw (45) to (47.center);
		\draw (51) to (44);
		\draw (51) to (46);
		\draw (40) to (31);
		\draw (55) to (56);
		\draw (56) to (58.center);
		\draw (55) to (57.center);
		\draw (60) to (56);
		\draw (53) to (60);
		\draw (64) to (65);
		\draw (65) to (67.center);
		\draw (64) to (66.center);
		\draw (68) to (65);
		\draw (62) to (68);
		\draw (69) to (62);
		\draw (114) to (115);
		\draw (115) to (117.center);
		\draw (114) to (116.center);
		\draw (119) to (115);
		\draw (127) to (129.center);
		\draw (127) to (136);
		\draw (136) to (130.center);
		\draw (139) to (140.center);
		\draw (139) to (141.center);
	\end{pgfonlayer}
\end{tikzpicture}
\end{adjustbox}\vspace{2cm}
\begin{adjustbox}{width=\columnwidth-1cm,center}
    \begin{tikzpicture}
	\begin{pgfonlayer}{nodelayer}
		\node [style=white dot] (4) at (-14, 1.5) {};
		\node [style=white dot] (10) at (-17, -0.5) {};
		\node [style=dark] (11) at (-15, -0.5) {};
		\node [style=dark] (12) at (-11, -0.5) {};
		\node [style=white dot] (13) at (-13, -0.5) {};
		\node [style=dark] (14) at (-17, -2.5) {};
		\node [style=white dot] (15) at (-13, -2.5) {};
		\node [style=none] (16) at (-22, -2.5) {};
		\node [style=none] (17) at (-10, -2.5) {};
		\node [style=none] (18) at (-15, -0.5) {$\neg$};
		\node [style=none] (19) at (-11, -0.5) {$\neg$};
		\node [style=none] (20) at (-8, -0.75) {$=$};
		\node [style=none] (26) at (5, -0.75) {$+$};
		\node [style=none] (27) at (-6, -0.75) {$a_0$};
		\node [style=none] (28) at (7, -0.75) {$a_1$};
		\node [style=white dot] (29) at (0, 1.5) {};
		\node [style=white dot] (30) at (-3, -0.5) {};
		\node [style=dark] (31) at (-1, -0.5) {};
		\node [style=dark] (32) at (3, -0.5) {};
		\node [style=white dot] (33) at (1, -0.5) {};
		\node [style=dark] (34) at (-3, -2.5) {};
		\node [style=white dot] (35) at (1, -2.5) {};
		\node [style=none] (36) at (-5, -2.5) {};
		\node [style=none] (37) at (4, -2.5) {};
		\node [style=none] (38) at (-1, -0.5) {$\neg$};
		\node [style=none] (39) at (3, -0.5) {$\neg$};
		\node [style=dark] (40) at (-4, 1.5) {};
		\node [style=white dot] (41) at (13, 1.5) {};
		\node [style=white dot] (42) at (10, -0.5) {};
		\node [style=dark] (43) at (12, -0.5) {};
		\node [style=dark] (44) at (16, -0.5) {};
		\node [style=white dot] (45) at (14, -0.5) {};
		\node [style=dark] (46) at (10, -2.5) {};
		\node [style=white dot] (47) at (14, -2.5) {};
		\node [style=none] (48) at (8, -2.5) {};
		\node [style=none] (49) at (17, -2.5) {};
		\node [style=none] (50) at (12, -0.5) {$\neg$};
		\node [style=none] (51) at (16, -0.5) {$\neg$};
		\node [style=dark] (52) at (9, 1.5) {};
		\node [style=none] (53) at (9, 1.5) {$\neg$};
		\node [style=none] (55) at (5, -7) {$+$};
		\node [style=none] (56) at (-6, -7) {$a_0$};
		\node [style=none] (57) at (7, -7) {$a_1$};
		\node [style=white dot] (64) at (1, -7) {};
		\node [style=none] (65) at (-5, -7) {};
		\node [style=none] (66) at (4, -7) {};
		\node [style=dark] (75) at (10, -7) {};
		\node [style=none] (77) at (8, -7) {};
		\node [style=none] (78) at (17, -7) {};
		\node [style=none] (79) at (10, -7) {$\neg$};
		\node [style=none] (80) at (1, -7) {$\neg$};
		\node [style=white box] (81) at (-13, -1.5) {};
		\node [style=white box] (82) at (1, -1.5) {};
		\node [style=white box] (83) at (14, -1.5) {};
		\node [style=white box, minimum width=1cm, minimum height=2.5cm] (84) at (-19, 1.5) {$D_{\phi}$};
		\node [style=white dot] (85) at (-22, 3) {};
		\node [style=white dot] (86) at (-22, 0) {};
		\node [style=none] (87) at (-19, 3) {};
		\node [style=none] (88) at (-19, 0) {};
		\node [style=none] (89) at (-22, 1.75) {$\vdots$};
		\node [style=none] (90) at (-8, -7) {$=$};
	\end{pgfonlayer}
	\begin{pgfonlayer}{edgelayer}
		\draw (10) to (14);
		\draw (10) to (11);
		\draw (11) to (13);
		\draw (13) to (12);
		\draw (12) to (4);
		\draw (4) to (10);
		\draw (14) to (15);
		\draw (15) to (17.center);
		\draw (14) to (16.center);
		\draw (30) to (34);
		\draw (30) to (31);
		\draw (31) to (33);
		\draw (33) to (32);
		\draw (32) to (29);
		\draw (29) to (30);
		\draw (34) to (35);
		\draw (35) to (37.center);
		\draw (34) to (36.center);
		\draw (40) to (29);
		\draw (42) to (46);
		\draw (42) to (43);
		\draw (43) to (45);
		\draw (45) to (44);
		\draw (44) to (41);
		\draw (41) to (42);
		\draw (46) to (47);
		\draw (47) to (49.center);
		\draw (46) to (48.center);
		\draw (52) to (41);
		\draw (64) to (66.center);
		\draw (75) to (77.center);
		\draw (65.center) to (80.center);
		\draw (79.center) to (78.center);
		\draw (81) to (13);
		\draw (81) to (15);
		\draw (82) to (33);
		\draw (82) to (35);
		\draw (83) to (45);
		\draw (83) to (47);
		\draw (85) to (87.center);
		\draw (86) to (88.center);
		\draw (84) to (4);
	\end{pgfonlayer}
\end{tikzpicture}
\end{adjustbox}
\end{center}

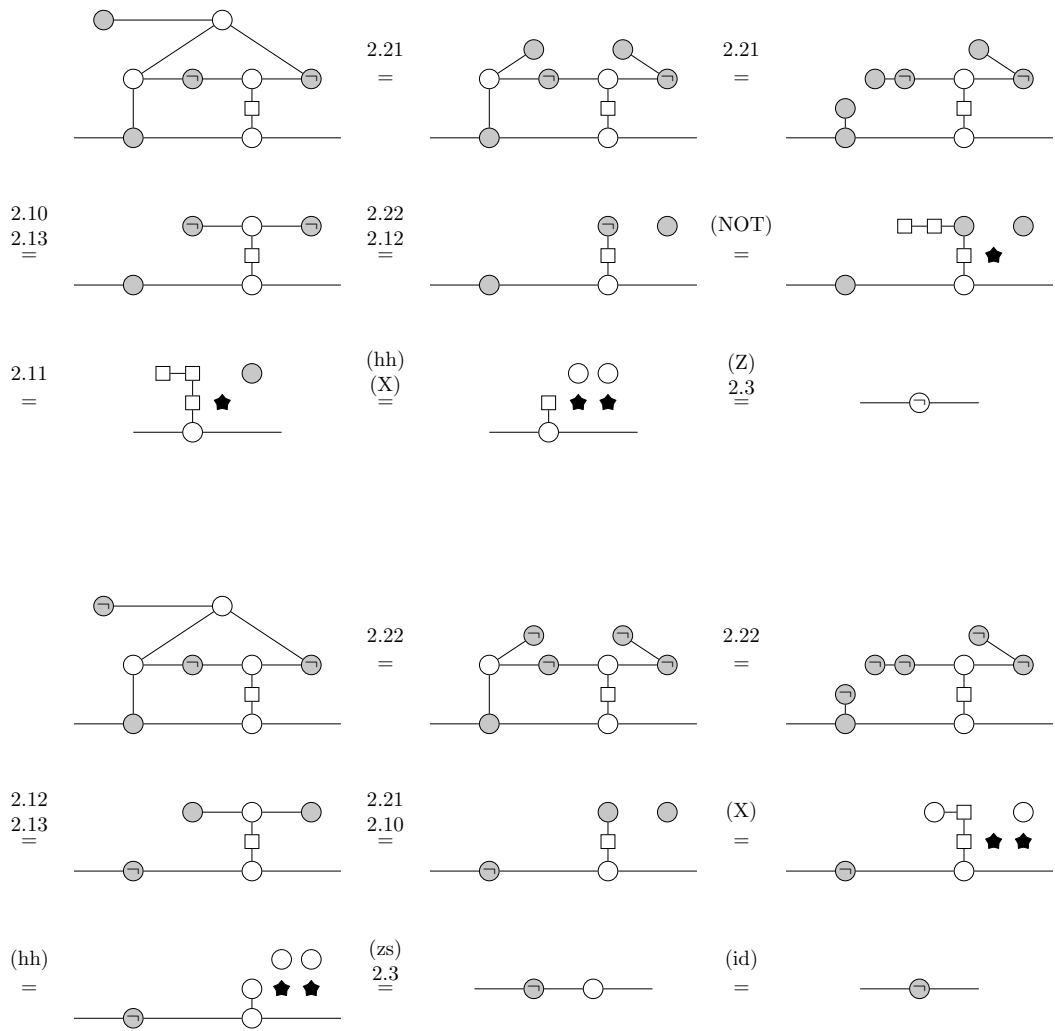
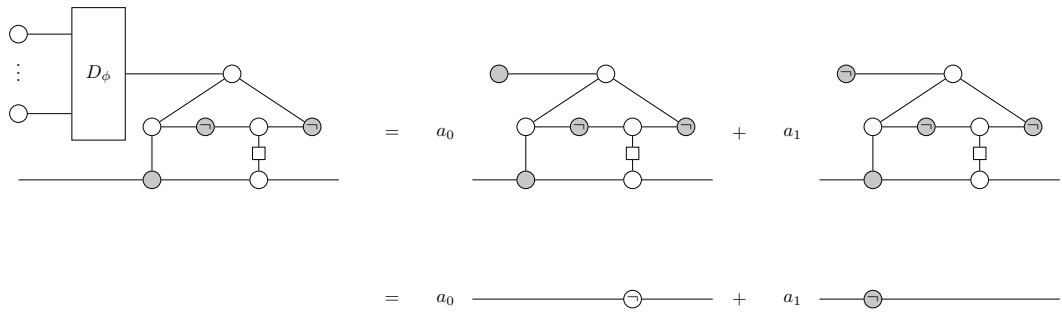
\end{document}